\documentclass[twocolumn,prx,longbibliography,nofootinbib, superscriptaddress]{revtex4-2}

\usepackage[dvips]{graphicx} 
\usepackage{amsfonts}
\usepackage{amssymb}
\usepackage{amscd}
\usepackage{amsmath}    
\usepackage{amsthm}
\usepackage{appendix}
\usepackage{bbm}
\usepackage{dsfont}
\usepackage{enumerate}
\usepackage{epsfig}
\usepackage{float}
\usepackage[colorlinks, linkcolor=red, anchorcolor=blue, citecolor=green]{hyperref}
\usepackage{makecell}
\usepackage{subfigure}
\usepackage{subfloat}
\usepackage{xcolor}		
\usepackage{physics}
\usepackage[most]{tcolorbox}
\usepackage{tabularx}

\usepackage{MnSymbol}
\usepackage{tikz}
\usepackage{tikzpeople}
\usepackage{pgfplots}
\usetikzlibrary{quantikz}
\usetikzlibrary{arrows}
\usetikzlibrary{shapes,fadings,snakes}
\usetikzlibrary{decorations.pathmorphing,patterns}
\usetikzlibrary{calc}
\usetikzlibrary{positioning}
\definecolor{amber}{rgb}{1.0, 0.75, 0.0}

\newtheorem{theorem}{Theorem}
\newtheorem{fact}{Fact}
\newtheorem{lemma}{Lemma}

\newtheorem{definition}{Definition}
\newtheorem{proposition}{Proposition}

\graphicspath{{./figure/}}

\newcommand{\dtr}{\text{d}_{\tr}}
\newcommand{\bE}{\mathbb{E}}

\newcommand{\bR}{\mathbb{R}}
\newcommand{\cU}{\mathcal{U}}
\newcommand{\cV}{\mathcal{V}}
\newcommand{\cW}{\mathcal{W}}

\newcommand{\cC}{\mathcal{C}}
\newcommand{\cS}{\mathcal{S}}
\newcommand{\cO}{\mathcal{O}}
\newcommand{\cK}{\mathcal{K}}
\newcommand{\cE}{\mathcal{E}}
\newcommand{\ba}{\mathbf{a}}
\newcommand{\bb}{\mathbf{b}}

\newcommand{\poly}{\text{poly}}

\renewcommand{\norm}[1]{\left\lVert#1\right\rVert}
\newcommand{\dTV}[1]{\norm{#1}_{\mathrm{TV}}}

\newtcolorbox[auto counter]{mybox}[2][]{
	enhanced,
	breakable,
	colback=blue!5!white,
	colframe=blue!75!black,
	fonttitle=\bfseries,
	title=Box \thetcbcounter: #2,#1
}

\begin{document}
	\title{Spacetime Quantum Circuit Complexity via Measurements}
	
	\author{Zhenyu Du}
	\affiliation{Center for Quantum Information, Institute for Interdisciplinary Information Sciences, Tsinghua University, Beijing 100084, China}
	\author{Zi-Wen Liu}
	\email{zwliu0@tsinghua.edu.cn}
	\affiliation{Yau Mathematical Sciences Center, Tsinghua University, Beijing 100084, China}
	
	\author{Xiongfeng Ma}
	\email{xma@tsinghua.edu.cn}
	\affiliation{Center for Quantum Information, Institute for Interdisciplinary Information Sciences, Tsinghua University, Beijing 100084, China}
	
	\begin{abstract}
		Quantum circuit complexity is a fundamental concept whose importance permeates quantum information, computation, many-body physics and high-energy physics. While extensively studied in closed systems, its characterization and behaviors in the widely important setting where the system is embedded within a larger one---encompassing measurement-assisted state preparation---lack systematic understanding. 
		We introduce the notion of embedded complexity that characterizes the complexity of projected states and measurement operators in this general setting incorporating auxiliary systems and measurements. 
		For random circuits and certain strongly interacting time-independent Hamiltonian dynamics, we show that the embedded complexity is lower-bounded by the circuit volume---the total number of gates acting on both the subsystem and its complement.
		This strengthens the complexity linear growth theorems, enriches the understanding of deep thermalization, and indicates that measurement-assisted methods generically cannot yield significant advantages in state preparation cost, contrary to expectations.
		We further demonstrate a spacetime conversion of certain circuit models that concentrates circuit volume onto a subsystem, and showcase applications for random circuit sampling and shadow tomography.
		Our theory establishes a unified framework for space and time aspects of quantum circuit complexity, yielding profound new insights and applications across quantum information and physics.
	\end{abstract}
	
	\maketitle
	
	Defined as the minimal number of local gates required to generate a state or evolution, quantum circuit complexity holds pivotal importance across various domains ranging from quantum information \cite{aaronson2016complexity, bouland2019computational, Brand_o_2021} to physics \cite{Stanford_2014,susskind2014computational, Brown_2016, Brown_2018, susskind2018black,wen2013topological,Yi_2024_complexity}. In sharp contrast to usual properties such as entanglement \cite{Nahum_2017} which are bounded by system size, the circuit complexity of a quantum circuit can grow with the circuit depth to reach values exponential in system size \cite{Haferkamp_2022_linear}.  This provides a novel lens on the prolonged evolution in a closed system, with deep connections to holography and high-energy physics in the context of the AdS/CFT correspondence~\cite{Stanford_2014,susskind2014computational, Brown_2016, Brown_2018, susskind2018black}. It is also  crucial in quantum many-body physics, underpinning the theory of quantum phases of matter and topological order~\cite{Chen2010LongRangeEntanglement, wen2013topological}.
	
	Beyond closed systems, understanding the properties of a subsystem embedded in a larger system is a widely important problem in many-body physics, crucial for a deep understanding of phenomena including thermalization \cite{Rigol_2008, Kaufman_2016, Abanin_2019_manybody} and quantum chaos \cite{D_Alessio_2016, Ho_2022, Ippoliti_2023}. 
	In many-body quantum systems, a subsystem may exhibit significant entanglement with its extensive complement \cite{Nahum_2018}. Such entanglement behavior is closely relevant to information scrambling \cite{Mi_2021} and quantum error correction~\cite{Choi_2020, Yi_2024_complexity}.
	Moreover, after polynomial time evolution, universal and highly random quantum state ensembles within a subsystem can be encoded in a single state of a large system \cite{Cotler_2023, Choi_2023}, signaling the complex and rich properties of subsystems over extended durations.

	Notably, the surging interest in utilizing measurements to manipulate and understand subsystems, paralleled by experimental progress across various platforms~\cite{Arute_2019, Ni_2023, cai2023protecting, Evered_2023, dasilva2024demonstration}, has driven numerous important developments.
	A fundamental phenomenon known as deep thermalization concerns higher moments of projected ensembles within a subsystem induced by projective measurements~\cite{Ho_2022, Claeys2022dualunitary, Choi_2023, Bhore_2023, Cotler_2023, Ippoliti_2023, chang2024deepthermalizationchargeconservingquantum}, revealing physics beyond conventional thermalization and entanglement with significance in both theory~\cite{mark2024maximum, mok2024optimalconversionclassicalquantum} and practical protocols such as benchmarking \cite{Choi_2023} and shadow tomography \cite{Tran_2023}.
	Another insight that has generated wide interest and applications in quantum computing and physics is that measurements can significantly enhance entanglement, offering shortcuts for generating important quantum systems associated with e.g.~topologically ordered phases and quantum error-correcting codes \cite{Briegel01,verresen2022efficiently, bravyi2022adaptive, Lu_2022, Tantivasadakarn_2023, Lu_2023, Tantivasadakarn_2024, Piroli_2021, bravyi2022adaptive}. 
	
	These broad perspectives together signal the importance of understanding the complexity of quantum operations and states in the measurement-projected setting.
	Here we address this by introducing embedded complexity, a unifying extension of traditional closed-system circuit complexity, to encompass ancillae and measurements which essentially mediate between space and time resources. 
	We establish rigorous connections between the embedded complexity and quantum circuit volume which capture the total gate cost across both the subsystem and the ancillary system, in both local quantum circuits and Hamiltonian evolution settings.
	As we will elaborate, this yields a fundamental generalization of the complexity linear growth phenomenon~\cite{Brand_o_2021, Haferkamp_2022_linear, li2022short, chen2024incompressibility} to spacetime, and advances our understanding of deep thermalization and the limitation of measurement-assisted state preparation.
	We further establish a spacetime conversion for random and Clifford circuits through protocols that use measurements to trade ancillary qubits for circuit depth. 
	We showcase the practical utility of our protocols in two important applications: random circuit sampling and shadow tomography.

	\textit{Key definitions}---The conventional quantum circuit complexity $C$ is defined as the minimal number of local unitary gates (without loss of generality, 2-local untaries from $SU(4)$ acting on any two sites) required to generate the state or implement the operator across all possible circuits \cite{footnote:native_gate}. 
	Incorporating state preparation utilizing ancillas and mid-circuit measurements~\cite{Piroli_2021, verresen2022efficiently, bravyi2022adaptive, Lu_2022, Tantivasadakarn_2023, Lu_2023, Tantivasadakarn_2024}, we define the spacetime version that we dub \emph{embedded complexity} as follows.
	
	\begin{definition} [Embedded complexity] \label{def:C_anc_state}
		The embedded complexity $C_{anc}(\ket{\psi})$ of a pure $n$-qubit state $\ket{\psi}$ is defined as the minimal number of 2-qubit gates required to generate $\ket{\psi}$ within an $n$-qubit subsystem embedded in a $m$-qubit larger system. Single-qubit computational-basis measurements and post-selection are allowed in the middle of the circuit:
		\begin{equation}
			\begin{aligned}
				C_{anc}(\ket{\psi}) \coloneqq \min \{V: \exists &m\ge n, c > 0, \ket{\psi} \otimes \ket{0}^{\otimes (m - n)}  = \\
				& c \Pi_V U_V \Pi_{V-1}U_{V-1} \cdots \Pi_1 U_1 \ket{0}^{\otimes m}\}
			\end{aligned} 
		\end{equation}
		The 2-qubit gates $U_i$ can be arbitrary unitaries in $SU(4)$ and may act on any pair of qubits. The projective operator $\Pi_i$ acts on the same pair of qubits as $U_i$,
		\begin{equation}
			\begin{split}
				&\Pi_i = P_{i,1} \otimes P_{i,2}, \\
				&P_{i,1}, P_{i,2} \in \{I, \ketbra{0}, \ketbra{1}\}.
			\end{split}
		\end{equation}
	\end{definition}
	
	An analogous definition of embedded complexity for Kraus operators is presented in Appendix~\ref{subsec:complexity_Kraus}. In defining the embedded complexity, $\Pi_i$ operators capture mid-circuit measurements and post-selections, and $c$ is the normalization factor.  
	Each unitary $U_i$ may depend on the outcomes of prior projective measurements, allowing for adaptive operations based on earlier measurement results.
	Therefore, embedded complexity characterizes the optimal resources needed for measurement-assisted state preparation protocols. 
	Moreover, it is evident that the embedded complexity lower-bounds the conventional complexity $C$ (which only involve $U_i$'s) \cite{Haferkamp_2022_linear}.
	It is also standard to introduce \emph{approximate embedded complexity} as a robust notion that incorporates error tolerance by a further optimization over all states within a certain distance from the target. In the main text, to highlight the essence of our results, we omit error dependence and use the notation $\widetilde{C}_{anc}$ to loosely denote the approximate embedded complexity given by an arbitrary finite universal gate set; all detailed definitions and results can be found in Appendices~\ref{app:approx_embedded} and \ref{app:approx_embed_Hamiltonian}.
	
	The \emph{circuit volume} $V$ is defined as total number of 2-qubit gates in a specific preparation process, which  quantifies the total spacetime cost. 
	For local circuit models, as shown in Fig.~\ref{fig:complexity_volume}(b), an $m$-qubit local circuit $U$ with depth $d$ is constructed as $U = U^{(d)}U^{(d-1)}\cdots U^{(1)}$, where $U^{(1)} = U_{1,2}^{(1)} \otimes U^{(1)}_{3,4}\otimes \cdots $ and subsequent layers $U^{(2)}, U^{(3)}, \ldots, U^{(d)}$ follow $U^{(1)}$ in a staggered arrangement. Each 2-qubit gate $U^{(i)}_{j,j+1}$ in the $i$-th layer acts on qubits $j$ and $j+1$. 
	Then, the circuit volume is given by $V = \lfloor m / 2 \rfloor d$. For a time-evolution $e^{-iH\tau}$ under a Hamiltonian $H$ with properly normalized local terms, the circuit volume is defined as $V = m\tau$.

	\begin{figure}[bthp!]
		\includegraphics[scale=0.75]{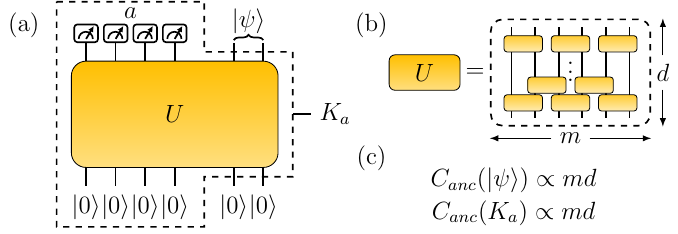}
		
		\caption{Embedded complexity and quantum circuit volume. (a) We study the embedded complexity of projected states and Kraus operators in a small subsystem obtained by applying a quantum circuit $U$ to the {all-zero initial} state and performing local projective measurements on its complement. (b) For local random circuit model, the unitary $U$ is randomly drawn from the local random circuits ensemble $\cU_{m,d}$ on $m$ qubits with circuit depth $d$. (c) Theorem \ref{thm:complexity_volume} show that with unit probability, the embedded complexity of the projected state $\ket{\psi}$ in the small subsystem and the Kraus operator $K_a$ is lower-bounded by the circuit volume. }
		\label{fig:complexity_volume}
	\end{figure}

	\textit{Bounding embedded complexity by circuit volume}---We study the embedded complexity of projected states prepared on an $n$‑qubit subsystem by performing local projective measurements on the complementary subsystem of a larger $m$‑qubit system, as depicted in Fig.~\ref{fig:complexity_volume}(a). Upon obtaining a measurement outcome $a \in \{0,1\}^{m-n}$, the prepared projected state is
	\begin{equation}
		\ket{\psi} \propto (\bra{a} \otimes I_n) U \ket{0}^{\otimes m}.
	\end{equation}
	The measurement also performs a POVM on the $n$‑qubit subsystem, where each outcome $a\in\{0,1\}^{m-n}$ on the ancillary qubits corresponds to the Kraus operator $K_a = (\bra{a} \otimes I_n) U (\ket{0}^{\otimes (m-n)} \otimes I_n)$.
	
	We first analyze the canonical local random circuit model \cite{Brand_o_2016, Haferkamp_2022_linear, Haferkamp2022randomquantum, chen2024incompressibility} to understand the typical behaviors of embedded complexity. Our method is extendable to higher dimensions and various architectures.  Here, each 2-qubit gate is independently drawn from the Haar measure on $SU(4)$. We denote by $\mathcal{U}_{m,d}$ the ensemble of $m$-qubit local random circuits with depth $d$.
	The corresponding ensemble of quantum states, obtained by applying a unitary $U \in \mathcal{U}_{m,d}$ to the all-zero initial state, is defined as
	\begin{equation}
		\mathcal{S}_{m,d} = \{U\ket{0}^{\otimes m}: U \in \mathcal{U}_{m,d}\}.
	\end{equation}
	The probability distributions over both $\mathcal{U}_{m,d}$ and $\mathcal{S}_{m,d}$ are induced by the Haar measure over the individual 2-qubit gates.
	
	To determine the embedded complexity of the projected state or Kraus operator, one must take minimization over all viable circuits and measurements. 
	This task is notoriously challenging due to the difficulty in conclusively eliminating the possibility of reducing the number of gates. 
	A reduction in complexity seems especially possible when the final $n$-qubit subsystem is much smaller than the initial $m$-qubit system (i.e., $n \ll m$), as most two-qubit gates lie outside the lightcone of the subsystem. 
	Remarkably, we prove that the circuit volume is nearly incompressible. 
	
	\begin{theorem}\label{thm:complexity_volume}
		Given $m \ge n \ge 4$, consider a local random circuit $U \in \cU_{m,d}$ acting on the initial state $\ket{0}^{\otimes m}$. After the first $m-n$ qubits of the state $U\ket{0}^{\otimes m}$ are measured in the computational basis, the projected state $\ket{\psi}$ on the remaining $n$ qubits will, with unit probability, satisfy:
		\begin{equation}
			C_{anc}(\ket{\psi}) \ge \min\left(\frac{md}{2n^2} - 2m, 2^{n+1}-2\right) / 15.
		\end{equation}
		For $d = \Omega(n^2)$, the bound can be made $C_{anc}(\ket{\psi}) = \Omega(\min(\frac{V}{n^2}, 2^n))$, where $V = \lfloor m / 2 \rfloor d$ is the circuit volume.
	\end{theorem} 
	
	We summarize this theorem in Fig.~\ref{fig:complexity_volume}(c).
	This theorem implies that, in almost all but extremely special cases, the use of ancillas and measurements does not permit a substantial reduction of the circuit volume $V$ (only scaled by a factor of $O(n^{-2})$).
	Within an $n$-qubit closed system, the projected states thus require preparation time at least
	$V/\mathrm{poly}(n)$, which can far exceed the original depth $d$ in the $m$-qubit system when
	$n\ll m$, revealing a spacetime tradeoff of circuit complexity.
	
	Underpinning this result are two key insights: 
	(i) Rather than merely destroying entanglement~\cite{Brian2019MeasurementInduced}, the measurements performed after deep circuits concentrate the degrees of freedom from the measured ancillary qubits into the unmeasured subsystem.
	We show that the projected states obtained from local random circuits form high‑dimensional manifolds, whose dimensions scale proportionally with the circuit volume, which is rigorously characterized using tools from semi‑algebraic geometry~\cite{Haferkamp_2022_linear, li2022short, suzuki2023quantum};
	(ii) Measurements performed within low‑complexity circuits do not increase the dimension of the manifolds of preparable states, due to the finiteness of the measurement outcomes.
	Combining these two insights, we show that measurement‑assisted quantum circuits can only reach a measure‑zero subset of the full projected state manifolds.
	A complete proof, together with analogous results for Kraus operators, is provided in Appendix~\ref{app:bounding}.

	Further, we establish that measurements cannot significantly simplify the generation of designs (statistically pseudorandom ensembles that reproduce the uniform Haar measure up to certain moments), a paradigmatic practical notion of randomness that  naturally emerge from e.g., random circuit~\cite{Brand_o_2016} and chaotic Hamiltonian dynamics~\cite{Nakata_2017, Ho_2022, Ippoliti_2023} and holds fundamental importance across quantum information and physics. As an example of its application, the proof of Theorem~\ref{thm:hamiltonian} uses this result.
	
	\begin{proposition}[Informal]\label{lem:design_approximate_complexity}
		Let $\ket{\psi}$ be an $n$-qubit state sampled from an approximate state $k$-design with $k < 2^{n/2}$. Then, with high probability, $\widetilde{C}_{anc}(\ket{\psi}) = \Omega(nk)$.
	\end{proposition}
	
	Extending beyond random states, we also consider  Hamiltonian dynamics and show that the relation between embedded complexity and circuit volume holds for projected states produced by a time‑independent Hamiltonian evolution. 
	As a concrete example, consider a two‑dimensional lattice of $m_r \times m_c$ qubits with local Hamiltonian $H = \sum_i h_i X_i + \sum_{i,j} h_{i,j} X_iX_j $,
	where the on‑site fields $h_i$ and interaction strengths $h_{i,j}$ are specified in Appendix~\ref{app:approx_embed_Hamiltonian}. 
	We study the projected state on a single-column subsystem, as shown in Fig.~\ref{fig:qresource}(a).
	The following theorem parallels our result for random circuits (Theorem~\ref{thm:complexity_volume}) and confirms that the complexity of the projected state is lower-bounded by the circuit volume of Hamiltonian evolution. This suggests a spacetime conversion for circuit complexity in time-independent Hamiltonian dynamics.
	The proof combines measurement-based (MB) protocols in Refs.~\cite{Turner_2016,Mezher2018GraphStates} with our Proposition~\ref{lem:design_approximate_complexity} (see Appendix \ref{app:approx_embed_Hamiltonian} for details).
	
	\begin{theorem}[Informal]\label{thm:approximate_embedded_complexity_Hamiltonian}
		Consider {the above local Hamiltonian defined on a} two-dimensional $m_r \times m_c$ lattice. There exists an evolution time $\tau$ such that, after measuring $m_r(m_c-1)$ qubits of the evolved state $\exp(-iH\tau) \ket{0}^{\otimes m_rm_c}$ in the computational basis, the projected state $\ket{\psi}$ on the $m_r$ qubits in the last column with high probability satisfies
		\begin{equation}
			\widetilde{C}_{anc}(\ket{\psi}) \ge \min\left(\frac{V}{\mathrm{poly}(m_r)}, 2^{\Omega(m_r)}\right),
		\end{equation}
		where $V = m_r m_c \tau$ is the circuit volume.
		\label{thm:hamiltonian}
	\end{theorem}
	
	\begin{figure}[!h]
		\includegraphics[width=0.45\textwidth]{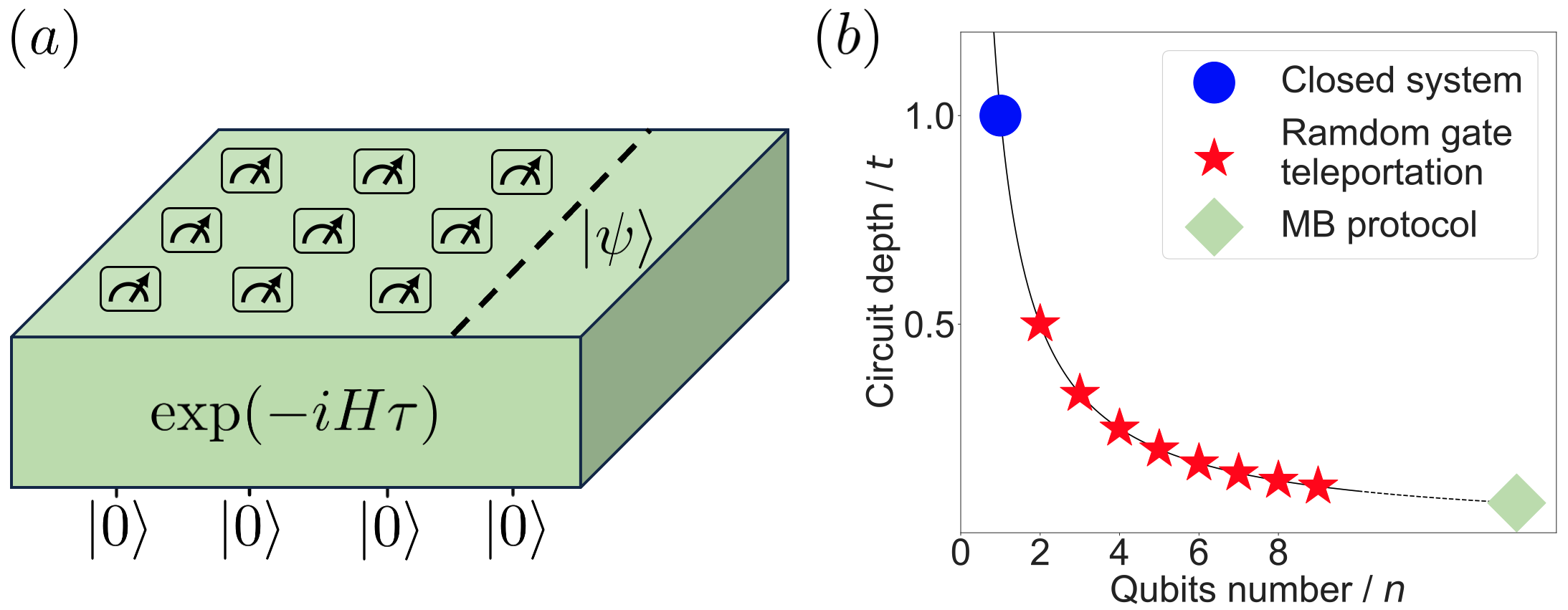}
		\caption{(a) Hamiltonian evolution in Theorem~\ref{thm:hamiltonian}. Measurements are performed on the first $m_r(m_c-1)$ qubits, leaving a projected state on the final column. (b) Spacetime complexity conversion in random and Clifford circuits. This diagram illustrates the tunable tradeoff between the number of qubits $n$ and (relative) circuit depth $t$ enabled by our gate teleportation protocols, which interpolates between closed-system circuits and MB protocols. 
		}
		\label{fig:qresource}
	\end{figure}
	
	Our projected state setting and complexity results offer insights for more fields across physics and quantum information, which we now exemplify.

	Recent studies show that measurements can help shortcut the preparation of certain highly structured states including various paradigmatic entangled states and topologically ordered states~\cite{Briegel01,Piroli_2021,Lu_2022,Lu_2023,bravyi2022adaptive,Tantivasadakarn_2023}.
	Our incompressibility results indicate that measurement-assisted circuits do not significantly enlarge the set of preparable states and thus offer no nontrivial shortcuts for generic states, substantiating the insight that the advantages of measurement-assisted preparation are very rare and hinge on highly tailored structures.
	
	In quantum gravity, an influential proposal of Brown and Susskind originated from holographic insights~\cite{Brown_2018, susskind2018black} posits that the circuit complexity of generic physical dynamics grows linearly for exponentially long time.
	While recent progress has validated this linear growth conjecture to varying extents in random circuit models~\cite{Haferkamp_2022_linear,chen2024incompressibility}, existing understanding is limited to the basic closed-system unitary evolution scenario, with fundamental elements of quantum physics including spacetime, measurements and open-system dynamics have yet to enter the picture.
	Our embedded complexity bounds imply generalized linear growth theorems for spacetime complexity that unify these elusive aspects, strengthening our understanding of circuit complexity as a crucial lens into quantum gravity~\cite{Stanford_2014,susskind2014computational,Brown_2018, susskind2018black,bouland2019computational,jian2023subsystem}.
	
	Furthermore, the behaviors of projected states and ensembles are of wide importance in quantum many-body and statistical physics. Recently it has been recognized that they provide new insights into non-equilibrium physics, spawning active areas like deep thermalization and emergent randomness~\cite{Ho_2022, Claeys2022dualunitary, Bhore_2023, Cotler_2023, Ippoliti_2023, mark2024maximum}. Our embedded complexity theory further expands the intensively studied connection between complexity and scrambling physics~\cite{Brown_2018, susskind2018black,Roberts_2017,Liu_2018,Brand_o_2021,Haferkamp_2022_linear}.
	For instance, it enriches our understanding of deep thermalization by strengthening the known state design characterizations: as discussed, our theorems above reveal a fundamental complexity concentration phenomenon upon measurements and indicate that the projected states exhibit circuit complexity proportional to the circuit volume, generally far exceeding what the traditional state-design arguments would suggest~\cite{Brand_o_2021,Cotler_2023}.

	\textit{Spacetime conversion in quantum circuits}---Spacetime tradeoffs in quantum circuits are  crucial to quantum computing \cite{Bravyi2016trading, Maslov2021quantumadvantage, Zhang2022StatePreparation}, paralleling its long-standing interest in classical complexity theory \cite{Arya2009SpacetimeSearching, Paul2003TimespaceRamdomizedComputation}. 
	We devise explicit protocols that realize spacetime conversions for two paradigmatic circuit families—random circuits and Clifford circuits. 
	
	Our main technique is quantum circuit teleportation between subsystems by Bell state measurement. Informally, one prepares the Choi states of quantum circuits in different subsystems, then performs Bell state measurements to concentrate all circuits in a small subsystem. While gate teleportation may introduce Pauli gates interleaved within the circuits, these gates do not affect the specific circuits we aim to implement.
	Specifically, when choosing the unitaries $U_i$ as local random circuits, the Pauli gates can be absorbed into the random circuits $U_i$ thanks to the property of the Haar measure.  This allows us to obtain a random circuit in a subsystem with increased circuit depth. 
	The result is summarized below and illustrated in Fig.~\ref{fig:space_to_time}. Similar spacetime conversions for Clifford circuits and  stabilizer state preparation are detailed in Appendix~\ref{app:spacetime}.
	
	\begin{figure}[!htbp]
		\includegraphics[scale=0.75]{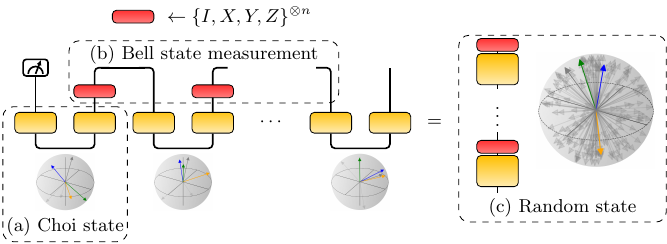}
		\caption{Spacetime conversion for random circuits. (a) First prepare Choi states of random circuits in different subsystems with circuit depth $d$. (b) Then perform Bell state measurements  to teleport random circuits (yellow rectangle) from one subsystem to another, which may introduce Pauli gates $P$ (red rectangle). $P$ is uniformly distributed on $\{I,X,Y,Z\}^{\otimes n}$, which can be absorbed into the random gates. (c) Finally, a random state in $\cS_{n,t}$ is prepared with reduced circuit depth $d$.}
		\label{fig:space_to_time}
	\end{figure}
	
	\begin{theorem}[Spacetime conversion for random circuits]
		\label{thm:random_small_depth}
		Given a circuit depth $t$, for an integer $k \ge 2$, quantum circuits on $kn$ qubits with a reduced depth $d = \left\lfloor \frac{t}{k} \right\rfloor + 4$ is sufficient for i) generating a random state in $\mathcal{S}_{n,t}$; ii) applying a random gate $U$ from $\mathcal{U}_{n,t}$ to any input state $\ket{\phi}$ when $k$ is an odd number.
	\end{theorem}
	
	Our protocols bridge closed-system and (constant-depth) MB schemes~\cite{Gottesman_1999,Raussendorf2001_oneway,Turner_2016}) for state generation, enabling a smoothly tunable spacetime resource conversion in between these two extremes (as illustrated in Fig.~\ref{fig:qresource}(b)). This offers practically precious flexibility in experiment and architecture design: using our method one can freely tailor the qubit and time costs to suit specific hardware features or capabilities. We briefly discuss applications in two particularly important scenarios (detailed discussion and results in Appendices~\ref{app:hardness_sampling} and \ref{app:shadow}).
	
	Random circuit sampling (RCS) is a flagship demonstration and benchmark for quantum advantage based on sampling from the distribution $p_U(x) = \abs{\bra{x}U\ket{0}}^2$ with $U$ drawn from certain random circuit ensemble~\cite{Arute_2019, Hangleiter2023AdvantageSampling, Liu2025CertifiedRandomness}.  
	We show that sampling from our random-gate-teleportation circuits remains classically hard under the same complexity assumptions as RCS of comparable circuit volume on a subsystem (Appendix \ref{app:hardness_sampling}), solidifying our message that the circuit volume establishes a fundamental spacetime characterization of the complexity of quantum systems. 
	This also rigorously validates the spacetime trade-off enabled by our methods in demonstrating quantum advantage, unlocking new routes for experiments as discussed earlier.
	
	Another notable application of spacetime conversion is in shadow tomography, where one aims to efficiently estimate certain properties of a state of interest. Existing protocols required evolving the input state online via random unitaries drawn from a unitary 3-design or Hamiltonian evolution \cite{Huang_2020, Tran_2023}. 
	In contrast, our protocol enables the simulation of random circuit action by preparing ancillary states and performing Bell measurements. 
	This approach is practically  appealing as it delegates the hardness from online (dynamics implementation) to offline (static state preparation), an
	insight that underpins many vital quantum computing schemes including magic state distillation~\cite{Bravyi_2005} and MBQC~\cite{Raussendorf2001_oneway,RaussendorfBrowneBriegel03,PRXQuantum.3.020333}. By applying Theorem~\ref{thm:random_small_depth} to prepare the required random ancilla states (from an approximate 3-design) in constant depth, our ancilla-assisted shadow tomography scheme can consequently predict global properties, such as fidelity with target states, using only constant-depth quantum circuits.
	
	\textit{Discussion}---We introduced and explored the concept of embedded complexity to incorporate measurements and space resource into characterization of circuit complexity.
	The connection we establish between the embedded complexity and circuit volume places fundamental limitations on what measurements and ancillary space can achieve, shedding light on holographic complexity, scrambling, and measurement-assisted dynamics, and is expected to extend to broader classes of physical systems.
	Conversely, certain quantum operations—such as quantum singular-value transformation \cite{Low2017QuantumSignalProcessing, Gilyen2019QSVT}—are difficult without ancillary qubits, since block encoding intrinsically requires them, underscoring the power of ancillary space and measurements. This raises an intriguing question: can measurement-assisted circuits provide an \emph{unconditional} gate-count advantage? For example, does there exist a class of quantum states $\{\ket{\psi_n}\}$ such that $C(\ket{\psi_n})=\omega\bigl(C_{\mathrm{anc}}(\ket{\psi_n})\bigr)$?
	
	A key application stemming from our framework is the spacetime circuit resource conversion.
	In the particularly important random circuit setting, this conversion indicates that the randomness of projected ensembles can be substantially enhanced by incorporating gate randomness, advancing previously studied settings using only measurements~\cite{Ho_2022,Ippoliti_2023,Cotler_2023} 
	or a single layer of random single-qubit gates~\cite{mok2024optimalconversionclassicalquantum}. 
	We believe further research into this randomness conversion and teleportation method would lead to abundant valuable advances in our understanding of the physics of complex quantum systems as well as technological applications including versatile methods for randomness generation \cite{ambainis2007quantum, Nakata_2017, Zhu_2017, chen2024efficient1,  metger2024simple, chen2024efficient2, mok2024optimalconversionclassicalquantum} with wide-ranging use in benchmarking \cite{Knill_2008, Alexander_2016, Mark_2023, Choi_2023}, compiling \cite{Wallman_2016, Hashim_2021}, learning \cite{Huang_2020, Tran_2023, McGinley_2023}, and beyond.
	
	\begin{acknowledgements}
		The authors thank Junjie Chen, Zhenhuan Liu, Yuxuan Yan, Xiao Yuan, and Qi Zhao for helpful discussion. This work is supported by the National Natural Science Foundation of China Grants No.~12174216 and the Innovation Program for Quantum Science and Technology Grant No.~2021ZD0300804 and No.~2021ZD0300702. Z.-W.L. is supported in part by a startup funding from YMSC, Tsinghua University, Dushi Program, and NSFC under Grant No.~12475023. 
	\end{acknowledgements}


	\bibliography{ref}

\begin{thebibliography}{95}
\bibitem [{foo()}]{footnote:native_gate}%
  \BibitemOpen
  \href@noop {} {}\bibinfo {note} {Here, the basic gate components are
  considered as two-qubit gates. We can also consider $k$-local gates as basic
  gates. Since any $k$-local gate can be decomposed into $\exp(O(k))$ two-qubit
  gates, this change in basic gates results in at most $\exp(O(k))$ scaling of
  state complexity, which remains a constant when $k$ is a
  constant.}\BibitemShut {Stop}%
\end{thebibliography}

	\clearpage
	\onecolumngrid
	
	\appendix
	\thispagestyle{empty}
	
	\renewcommand{\thetheorem}{S\arabic{theorem}}
	\renewcommand{\thefact}{S\arabic{fact}}
	\renewcommand{\thelemma}{S\arabic{lemma}}
	\renewcommand{\theequation}{S\arabic{equation}}
	\renewcommand{\thedefinition}{S\arabic{definition}}
	\renewcommand{\theproposition}{S\arabic{proposition}}
	
	\section{Spacetime conversion for random circuits and Clifford circuits}\label{app:spacetime}
	In this section, we provide the details of the spacetime conversion for random circuits and Clifford circuits. We begin by revisiting a gate teleportation protocol. Then, we prove Theorem 3 in the main text, which shows the spacetime conversion for random circuits. Finally, we extend the spacetime conversion result to Clifford circuits.
	\subsection{Preliminaries on gate teleportation}
	\subsubsection{Bell state measurement}
	Firstly, we revisit the notion of Bell state measurement, an important component in gate teleportation. Denote the unnormalized maximally entangled state as:
	\begin{equation}
		\ket{\Phi} = \ket{00} + \ket{11}. 
	\end{equation}
	We can represent $\ket{\Phi}$ using a tensor network diagram, as illustrated in Fig.~\ref{fig:EPR}(a).
	\begin{figure}[!h]
		\centering
		
		\subfigure[]{
			\raisebox{0.4cm}{ 
				\begin{minipage}[t]{0.4\linewidth}
					\centering
					
					\begin{tikzpicture}[
						scale=1,
						Lline/.style={-,black,thick,rounded corners=2},
						]
						\draw[Lline] (-1,0) -- (0,0) -- (0,-1)-- (-1,-1);
					\end{tikzpicture} 
					
				\end{minipage}
			}
		}
		\subfigure[]{
			\begin{minipage}[t]{0.4\linewidth}
				\centering
				
				\begin{tikzpicture}[
					scale=1,
					Lline/.style={-,black,thick,rounded corners=2},
					]
					
					\node[draw=black,rectangle,rounded corners=0.1cm,minimum width=.6cm,minimum height=.8cm,text=black,thick] (B) at (-1,0) {$U_A$};	
					\node[draw=black,rectangle,rounded corners=0.1cm,minimum width=.6cm,minimum height=.8cm,text=black,thick] (A) at ($(B) - (0,1)$) {$U_B$};	
					
					\draw[Lline] (B.east) -- (0,0) -- (0,-1) node [black,midway,xshift=.5cm] {$=$} -- (A.east);
					\draw[Lline] (B.west) --++ (-.5,0);
					\draw[Lline] (A.west) --++ (-.5,0);
					
					\node[draw=black,rectangle,rounded corners=0.1cm,minimum width=.6cm,minimum height=.8cm,text=black,thick] (B2) at (1.6, -1) {$U_B$};	
					\node[draw=black,rectangle,rounded corners=0.1cm,minimum width=.6cm,minimum height=.8cm,text=black,thick] (A2) at ($(B2) + (0.8,0)$) {$U_A^T$};
					
					\draw[Lline] (A2.east) -- (3,-1) -- (3,0) -- (1,0);
					\draw[Lline] (A2.west) -- (B2.east);
					\draw[Lline] (B2.west) -- (1,-1);
				\end{tikzpicture}
				
			\end{minipage} 
		}
		
		\caption{(a) Tensor network diagram of the unnormalized maximally entangled state $\ket{\Phi} = \ket{00} + \ket{11}$. (b) Diagram illustrating the movement of the unitary on the maximally entangled state. $U_A$ is applied on one side of the unnormalized maximally entangled state and can be moved to the other.}
		\label{fig:EPR}
	\end{figure}
	
	The four Bell states are abbreviated as:
	\begin{equation}\label{eq:Bell_state}
		\ket{\phi_{ab}} = \frac{1}{\sqrt{2}} (X^a Z^b \otimes I) \ket{\Phi}, \quad a,b\in \{0,1\}. 
	\end{equation}
	where $\ket{\phi_{00}}$ corresponds to the EPR pairs. We represent $\ket{\phi^n}_{AB}$ as the $n$ EPR pairs on systems $A$ and $B$:
	\begin{equation}
		\ket{\phi^n}_{AB} = \bigotimes_{i=1}^n \ket{\phi_{00}}_{A_i,B_i}, 
	\end{equation}
	where $A = A_1A_2\cdots A_n$ and $B = B_1B_2\cdots B_n$ are $n$-qubit system. For any $n$-qubit unitary $U$, we define the state $\ket{U,V}_{AB}$ as applying $U$ to $\ket{\phi^n}_{AB}$ on subsystem $A$ and $V$ on subsystem $B$:
	\begin{equation}
		\ket{U,V}_{AB} = (U \otimes V) \ket{\phi^n}_{AB}.
	\end{equation}
	
	A beneficial property is that one can move a unitary operation from one side of the maximally entangled state to the other:
	\begin{equation}\label{eq:move_unitary_EPR}
		\ket{U_A,U_B}_{AB} = \ket{I, U_BU^T_A}_{AB}
	\end{equation}
	where $U_A, U_B$ represents an $n$-qubit unitaries. The diagram representing Eq.~\eqref{eq:move_unitary_EPR} is shown in Fig.~\ref{fig:EPR}(b).
	
	We frequently utilize Bell state measurements in our analysis. Consider a $4n$-qubit quantum state $\ket{\psi_{ABCD}}$. Suppose Bell state measurements are performed on each $A_i$ and $C_i$ for $1 \leq i \leq n$ in the basis $\{\ket{\phi_{ab}}\}$, and the measurement outcome on the $i$-th pair of qubits is represented by $a_i$ and $b_i$, corresponding to the Bell state $\ket{X^{a_i}Z^{b_i},I}$. Now, let $\mathbf{a} = a_1a_2\cdots a_n$ and $\mathbf{b} = b_1b_2\cdots b_n$. Define $X^{\mathbf{a}} = X^{a_1} \otimes X^{a_2} \otimes \cdots \otimes X^{a_n}$ and $Z^{\mathbf{b}}$ analogously. The unnormalized post-measurement state on the system $BD$ after obtaining the measurement result $\mathbf{a}, \mathbf{b}$ is then given by:
	
	\begin{equation}\label{eq:post_measurement}
		(\ket{X^{\ba} Z^{\bb}, I}_{AC})^{\dagger} \otimes I_B \otimes I_D \ket{\psi_{ABCD}} = ( \bra{X^{\ba} Z^{\bb}, I}_{AC} \otimes I_B \otimes I_D) \ket{\psi_{ABCD}} 
	\end{equation}
	
	\subsubsection{Gate teleportation}
	In a seminal work on measurement-based quantum computing \cite{Gottesman_1999}, a construction for gate teleportation is proposed to apply a unitary $U$ to a state $\ket{\psi}$. The essence of this method is to perform Bell state measurements between a state $\ket{\psi}$ and the Choi state of a unitary $U$ instead of directly applying $U$ to $\ket{\psi}$. According to Eq.~\eqref{eq:post_measurement}, a Pauli error is applied before the unitary. This process is illustrated in Fig.~\ref{fig:tele_with_err} and summarized in the following lemma.
	\begin{figure}[!h]
		\centering
		\begin{tikzpicture}[
            Lline/.style={-,black,thick, rounded corners=2},
            ]
        \node[draw=black, isosceles triangle,
	isosceles triangle apex angle=90,
	draw,
	rotate=270,
	minimum size =.8cm] (T) at (0,0) {\rotatebox{90}{$\psi$}};

        \draw[Lline] ($(T.west)+(.5,0)$) --++ (0, 2) node (tmp) {};
        \draw[Lline] ($(T.west)-(.5,0)$) --++ (0, 2) --++ (-1,0) --++ (0,-.2) node[draw=black, rectangle, anchor=north, minimum width=0.5cm, minimum height=0.5cm, rounded corners=0.1cm, fill=red!50] (XZ) {$X^{\ba}Z^{\bb}$};
        \draw[Lline] (XZ.south) --++ (0,-1.8) --++ (-1,0) --++ (0,.5) node[draw=black, rectangle, anchor=south, minimum width=1cm, minimum height=1cm, rounded corners=0.1cm, fill=yellow!50] (U) {$U_B$};
        \draw[Lline] (U.north) -- (U.north |- tmp);

        \node (A) at ($(T.west) + (1.35,0.8)$) {=};
        \node[draw=black, isosceles triangle,
	isosceles triangle apex angle=90,
	draw,
	rotate=270,
	minimum size =.8cm] (T2) at ($(T.center) + (3.2,0)$) {\rotatebox{90}{$\psi$}};
        \draw[Lline] ($(T2.west)+(.5,0)$) --++ (0, 2);
        \draw[Lline] ($(T2.west)-(.5,0)$) --++ (0, 0.15) node[draw=black, rectangle, anchor=south, minimum width=0.5cm, minimum height=0.5cm, rounded corners=0.1cm, fill=red!50] (XZ2) {$X^{\ba}Z^{\bb}$};
        \draw[Lline] (XZ2.north) --++ (0,0.15) node[draw=black, rectangle, anchor=south, minimum width=1cm, minimum height=1cm, rounded corners=0.1cm, fill=yellow!50] (U2) {$U_B$};
        \draw[Lline] (U2.north) -- (U2.north |- tmp);
        \node (B0) at ($(U.south) + (0,-1.1)$) {$B$};
        \node (A0) at ($(B0) + (1,0)$) {$A$};
        \node (C0) at ($(B0) + (2,0)$) {$C$};
        \node (D0) at ($(B0) + (3,0)$) {$D$};
        \node (B01) at (U2 |- B0) {$B$};
        \node (D01) at ($(B01) + (1,0)$) {$D$};
\end{tikzpicture}
	
		\caption{Gate teleportation with Pauli error. Given two input state $\ket{I,U_B}_{AB}$ and $\ket{\psi}_{CD}$, the unitary applied on system $B$ can be teleported onto $\ket{\psi}$ via Bell state measurements, up to a Pauli error $P = X^{\mathbf{a}}Z^{\mathbf{b}}$, where $\mathbf{a}, \mathbf{b}$ are the results of the Bell state measurement.}
		\label{fig:tele_with_err}
	\end{figure}
	
	\begin{lemma}[Gate teleportation with Pauli error]
		\label{lem:teleport_pauli_error}
		Let $\ket{I,U_B}_{AB}$ and $\ket{\psi}_{CD}$ be two states, where $A, B, C, D$ are $n$-qubit systems. After performing Bell state measurements on subsystem $AC$ and obtaining measurement results $\mathbf{a}$ and $\mathbf{b}$, where $\mathbf{a}, \mathbf{b} \in \{0,1\}^n$, the resulting post-measurement state on subsystem $BD$ is $[(U_B P) \otimes I] \ket{\psi}_{BD}$, where $P = X^{\mathbf{a}} Z^{\mathbf{b}}$ is a Pauli gate.
	\end{lemma}
	
	For a Clifford unitary $V$, an adaptive Pauli operation $P'$ can correct the Pauli error $P$. This property arises from the ability to interchange a Pauli gate $P$ with a Clifford unitary $V$. Specifically, $VP = VPV^{\dagger}V = P'V$, where $P' = VPV^{\dagger}$ is also a Pauli gate due to the property of Clifford gates \cite{gottesman1998heisenberg}.
	
	\begin{lemma}[Clifford gate teleportation]
		\label{lem:teleport_clifford}
		Let $\ket{I,V}_{AB}$ and $\ket{\psi}_{CD}$ be two states, where $A,B,C,D$ are $n$-qubit systems and $V$ is an $n$-qubit Clifford gate. After performing Bell state measurements on subsystem $AC$ and obtaining the measurement results $\ba$ and  $\bb$, where $\ba,\bb \in \{0,1\}^n$, the resulting post-measurement state on subsystem $BD$ are $(P'V \otimes I) \ket{\psi}_{BD}$, where $P = X^{\ba}Z^{\bb}$ and $P'=VPV^{\dagger}$ are Pauli gates.
	\end{lemma} 
	
	\subsection{Spacetime conversion for random circuits} \label{app:add_proofs}
	
	Here, we provide detailed proof of Theorem 3, which shows the spacetime conversion for random circuits. First, we introduce the concepts of projected ensembles and construct a random gate teleportation protocol. Then, we prove the spacetime conversion for random circuits using this protocol.
	
	\subsubsection{Projected ensemble}
	Let us define the projected ensemble of a state ensemble after measuring a subsystem. 
	\begin{definition}[Projected ensemble of a state ensemble]
		For a state ensemble $\mathcal{S} = \{p_i, \ket{\psi_i}_{AB}\}$ on systems $A$ and $B$, the projected ensemble of $\mathcal{S}$ on subsystem $B$ after measuring subsystem $A$ in the basis $\{\ket{j_A}\}$ is denoted as 
		\begin{equation}
			\{p_{i}q_{ij}, \ket{\psi_{ij}} \},
		\end{equation} 
		where $q_{ij} = |(\langle j_A | \otimes I_B) \ket{\psi}|^2$ represents the probability of obtaining measurement result $j_A$ and  $\ket{\psi_{ij}} = \frac{(\langle j_A | \otimes I_B) \ket{\psi}}{\sqrt{q_{ij}}}$ represents the projected state on subsystem $A$. 
	\end{definition}
	
	Recall that $\ket{U,V}_{AB} = (U \otimes V) \ket{\phi^n}_{AB}$, where $A$ and $B$ are $n$-qubit quantum systems that are maximally entangled and $U,V$ are $n$-qubit unitaries acting on $A$ and $B$, respectively. We define the state ensembles by applying local random circuits on one side of the EPR pairs, i.e., the Choi states of local random circuits.
	
	\begin{definition}[Choi states of local random circuits]
		Let $A$ and $B$ be two $n$-qubit systems. The state ensemble $\mathcal{E}_{n,t}$ is defined as the set of states generated by applying a local random circuit $U \in \mathcal{U}_{n,t}$ to the subsystem $B$ of the EPR pairs $\ket{\phi^n}_{AB}$.
		\begin{equation}
			\mathcal{E}_{n,t} = \{\ket{I,U}_{AB} : U_B \in \mathcal{U}_{n,t}\}.
		\end{equation}
		This ensemble follows a probability distribution determined by $\mathcal{U}_{n,t}$.
	\end{definition}

	The projected ensemble of $\mathcal{E}_{n,t}$ on subsystem $B$ through computational-basis measurement on subsystem $A$ yields the state ensemble $\mathcal{S}_{n,t}$, a consequence of the local unitary invariance property of Haar measure on $SU(4)$. 
	
	\begin{lemma}\label{lem:proj_are_random}
		Consider the state ensembles $\mathcal{E}_{n,t}$ on systems $A$ and $B$. The projected ensembles of $\mathcal{E}_{n,t}$ on system $B$ through computational-basis measurement on subsystem $A$ is the state ensemble $\mathcal{S}_{n,t}$.
	\end{lemma}
	
	\begin{proof}
		For any state $\ket{\psi_{AB}} = \ket{I,U_B}_{AB}$, the reduced density matrix on subsystem $A$ is maximally mixed. Consequently, the measurement result $j$ on subsystem $A$ is uniformly distributed over $\{0,1\}^n$, and the corresponding post-measurement state on subsystem $B$ is $U_B \ket{j}_B$. Hence, the projected state ensemble is given by
		\begin{equation}
			\{U_B\ket{j}_B : U_B \in \mathcal{U}_{n,t}, j \in \{0,1\}^n\},
		\end{equation} 
		where the probability of $U_B$ is determined by $\mathcal{U}_{n,t}$, and $j$ is uniformly distributed over $\{0,1\}^n$. This state ensemble is exactly $\mathcal{S}_{n,t}$ due to the local unitary invariance property of Haar random distribution on $SU(4)$. 
	\end{proof}
	
	\subsubsection{Random gate teleportation}
	We show that two Choi states can be linked via Bell state measurements, resulting in Choi states of random circuits with increased circuit depth. This is achieved by teleporting the random circuits on one subsystem to another, a process we term random gate teleportation. This method is depicted in Fig.~\ref{fig:randomness_teleportation} and stated in the following lemma.
	
	\begin{figure}[!h]
		\centering
		\begin{tikzpicture}[
            scale=1,
            Lline/.style={-,black,thick, rounded corners=2},
            ]
        \coordinate (T) at (0,0);

        \draw[Lline] (T) --++ (0, -1)node[draw=black, rectangle, anchor=north, minimum width=1cm, minimum height=1cm, rounded corners=0.1cm, fill=yellow!50] (VD) {$V_D$};
        \draw[Lline] (VD.south)--++ (0,-1) --++ (-1,0) node (tmp) {} -- (tmp |- T)  --++(-1,0) --++ (0,-.2) node[draw=black, rectangle, anchor=north, minimum width=1cm, minimum height=0.5cm, rounded corners=0.1cm, fill=red!50] (XZ) {$P$};
        \draw[Lline] (XZ.south) -- (XZ |- tmp) --++ (-1,0) --++ (0,1) node[draw=black, rectangle, anchor=south, minimum width=1cm, minimum height=1cm, rounded corners=0.1cm, fill=yellow!50] (UB) {$U_B$};
        \draw[Lline] (UB.north) -- (UB |- T);
        \node (A) at ($(T) + (1,-1.4)$) {=};

        \coordinate (T2) at ($(T) + (3,0)$);
        \draw[Lline] (T2) --++ (0, -.1) node[draw=black, rectangle, anchor=north, minimum width=1cm, minimum height=1cm, rounded corners=0.1cm, fill=yellow!50] (VD2) {$V_D$};
        \draw[Lline] (VD2.south) --++ (0, -0.1) node[draw=black, rectangle, anchor=north, minimum width=1cm, minimum height=0.5cm, rounded corners=0.1cm, fill=red!50] (XZ2) {$P$};
        \draw[Lline] (XZ2.south) --++ (0,-0.1) node[draw=black, rectangle, anchor=north, minimum width=1cm, minimum height=1cm, rounded corners=0.1cm, fill=yellow!50] (UB2) {$U_B^T$};
        \draw[Lline] (UB2.south) -- (UB2 |- tmp) --++ (-1,0) node (tmp2) {} -- (tmp2 |- T2);
        
        \node (B0) at ($(UB.south) + (0,-1.5)$) {$B$};
        \node (A0) at ($(B0) + (1,0)$) {$A$};
        \node (C0) at ($(B0) + (2,0)$) {$C$};
        \node (D0) at ($(B0) + (3,0)$) {$D$};
        \node (D01) at (UB2 |- B0) {$D$};
        \node (B01) at ($(D01) + (-1,0)$) {$B$};
    \end{tikzpicture}
		\caption{Random gate teleportation. Systems $A$ and $B$ consist of $n$-qubit systems that are maximally entangled, with a local random circuit $U_B$ of depth $d_1$ applied to system $B$. Systems $C$ and $D$ share a similar configuration, with a local random circuit $V_D$ of depth $d_2$ on system $D$. The figure on the left-hand side shows the projected ensemble after performing Bell state measurements on subsystems $AC$, where $P$ is the Pauli error distributed uniformly over $\{I,X,Y,Z\}^{\otimes n}$. The projected ensemble is equivalent to the right-hand side, which are Choi states of random circuits with increased circuit depth.}
		\label{fig:randomness_teleportation}
	\end{figure}
	
	\begin{lemma}[Random gate teleportation]\label{lem:teleport_random_gate}
		Given state ensembles $\mathcal{E}_{n,d_1}$ on systems $AB$ and $\mathcal{E}_{n,d_2}$ on systems $CD$, where $A,B,C,$ and $D$ represent $n$-qubit quantum systems, through Bell state measurements on subsystems $AC$, the resulting projected ensemble on subsystems $BD$ is $\mathcal{E}_{n,d'}$ with $d' = d_1 + d_2 - 1$ or $d' = d_1 + d_2$.
	\end{lemma}

	\begin{proof}
		Consider any states $\ket{\psi}_{AB} = \ket{I,U_B}_{AB}$ and $\ket{\varphi}_{CD} = \ket{I,V_D}_{CD}$. The reduced density matrix of the joint system $\ket{\psi}_{AB} \otimes \ket{\varphi}_{CD}$, restricted to subsystem $AD$, is maximally mixed. As a result, the outcomes $\mathbf{a}$ and $\mathbf{b}$ from the Bell state measurements on these subsystems are uniformly distributed over $\{0,1\}^n$. 
		
		By applying Lemma \ref{lem:teleport_pauli_error}, the post-measurement state can be expressed as:
		\begin{equation}
			\begin{split}
				\ket{\psi_{\mathbf{a}\mathbf{b}}}_{BD} &= (U_B X^{\mathbf{a}} Z^{\mathbf{b}} \otimes I) \ket{I,V_D}_{BD} \\
				&= \ket{U_B X^{\mathbf{a}} Z^{\mathbf{b}}, V_D}_{BD} \\
				&= \ket{I, V_D Z^{\mathbf{b}} X^{\mathbf{a}} U_B^T}_{BD},
			\end{split}
		\end{equation}
		where the last equation is derived via Eq.~\eqref{eq:move_unitary_EPR}. Here, the unitaries $U_B$ and $V_D$ are sampled from $\mathcal{U}_{n,d_1}$ and $\mathcal{U}_{n,d_2}$, respectively. The bit-strings $\mathbf{a}$ and $\mathbf{b}$ are uniformly distributed over $\{0,1\}^n$.
		
		Due to the local unitary invariance and transpose-invariant properties of the Haar measure on $SU(4)$, $V_D Z^{\mathbf{b}} X^{\mathbf{a}}$ and $U_B^T$ are still local random circuits with unchanged circuit depth. Consequently, the composition $(V_D Z^{\mathbf{b}} X^{\mathbf{a}} U_B^T)$ results in a local random circuit of depth $d'$, where $d' = d_1 + d_2$ if the first layer of $V_D$ is staggered with the last layer of $U_B^T$, or $d' = d_1 + d_2 - 1$ if not.
	\end{proof}

	\subsubsection{Spacetime conversion for random circuits}
	
	By integrating Lemma \ref{lem:proj_are_random} with Lemma \ref{lem:teleport_random_gate}, we develop a method to construct state ensembles of local random circuits $\mathcal{S}_{n,t}$ using fewer layers of circuits via ancillary qubits and measurements. First, we show how to generate $\cE_{n,t}$ with reduced depth.
	
	\begin{lemma}\label{lem:random_EPR_small_depth}
		Given a circuit depth $t$, there exists a circuit depth $t_1 \ge t$ such that for any even integer $k = 2m$,  the state ensembles $\mathcal{E}_{n,t_1}$ can be generated within a total depth of $d = \left\lfloor \frac{t}{k} \right\rfloor + 4$. This is achieved by employing a random circuit on $kn$ qubits and performing Bell state measurements across $(k-2)n$ qubits.
	\end{lemma}

	\begin{proof}
		We partition $2mn$ qubits into $m$ blocks, each containing $2n$ qubits. On each pair of qubits within the blocks, we prepare EPR pairs and apply local random circuits in $\mathcal{U}_{n,d_2}$ on each side of EPR pairs separately, where $d_2 \ge \left\lfloor \frac{t}{2m} \right\rfloor + 2$. The state $\ket{\phi^n}$ within each block evolves to $\ket{U_1,U_2} = \ket{I,U_2U_1^T}$, with $U_1$ and $U_2$ drawn from $\mathcal{U}_{n,d_2}$. This state represents a member of the ensemble $\mathcal{E}_{n,d_3}$, where $d_3 \ge 2d_2 - 1$.
		
		By performing Bell state measurements to iteratively merge these blocks, we leverage Lemma \ref{lem:teleport_random_gate} to obtain a final state from the ensemble $\mathcal{E}_{n,t_1}$ in the last block, with $t_1 \ge d_3m - m \ge t$. These Bell state measurements can be performed simultaneously in a single layer. Therefore, the total circuit depth is $d = 1 + d_2 + 1 = \left\lfloor \frac{t}{2m} \right\rfloor + 4$. 
	\end{proof}
	
	Then, we can prove Theorem 3 in the main text, which shows a spacetime conversion for random circuits. 
	
	\begin{proof}[Proof of Theorem 3]
		For the first claim, when $k = 2m$, we apply Lemma \ref{lem:random_EPR_small_depth} to generate $\mathcal{E}_{n,t_1}$ on $kn$ qubits, where $t_1 \ge t$. Subsequently, a computational measurement is performed on one side of $\mathcal{E}_{n,t_1}$. According to Lemma \ref{lem:proj_are_random}, the projected ensemble is $\mathcal{S}_{n,t_1}$.
		
		When $k = 2m + 1$, we set $d_2 = \lfloor \frac{t}{k} \rfloor + 2$. Following the protocol in Lemma \ref{lem:random_EPR_small_depth}, we generate $\mathcal{E}_{n,d_3}$ with a depth of $d = d_2 + 2$ using $2mn$ qubits, where $d_3 \ge (2d_2 - 2)m$. Concurrently, we generate $\mathcal{S}_{n,d_2 + 1}$ on the remaining $n$ qubits. Bell state measurements are then performed on half of $\mathcal{E}_{n,d_3}$ and $\mathcal{S}_{n, d_2 + 1}$. According to Lemma \ref{lem:teleport_random_gate}, the projected ensemble is $\mathcal{S}_{n,t_1}$, with $t_1 \ge d_2 + d_3 \ge t$. These Bell state measurements can be executed simultaneously with the previous measurements. The total depth is $\lfloor \frac{t}{k} \rfloor + 4$.
		
		This procedure can be adapted to prove the second claim by applying a local random circuit of depth $d_2 + 1$ on the $n$-qubit state $\ket{\phi}$ instead of generating $\mathcal{S}_{n,d_2 + 1}$. The other steps are consistent with the proof of the first claim. 
		
		Finally, the circuit depth $t_1$ can be chosen equal to $t$ by appropriately arranging the random circuits, which proves the theorem.
	\end{proof}

	\subsection{Spacetime conversion for Clifford circuits}
	
	Similar to Theorem 3, a spacetime conversion for implementing Clifford circuits can be established by utilizing the Clifford gate teleportation protocol in Lemma \ref{lem:teleport_clifford}. We summarized this in the following theorem.
	
	\begin{theorem}	[Spacetime conversion for Clifford circuits]\label{thm:clifford_volume}
		Given an $n$-qubit Clifford circuit $C$ with circuit depth $t$, for an integer $k \ge 2$, quantum circuits on $kn$ qubits with a total depth $d = \left\lfloor \frac{t}{k} \right\rfloor + 4$ is sufficient to: 
		\begin{enumerate}
			\item {Prepare the state $C\ket{0}^{\otimes n}$. }
			\item {Apply the Clifford circuit $C$ to any input state $\ket{\phi}$ when $k$ is an odd number. }
		\end{enumerate} 
	\end{theorem}
	
	To establish this theorem, we first prove the following lemma.
	\begin{lemma}\label{lem:clifford}
		Consider an $n$-qubit Clifford circuit $C = C_{2m} C_{2m-1} \cdots C_1$, where each component $C_i$ is a Clifford circuit with a depth no greater than $d$. Using $(2m-2)n$ ancillary qubits, we can prepare the state $\ket{I, C}_{AB}$ with a total circuit depth $d + 2$, up to a Pauli error $P$ on subsystem $B$. Consequently, the output state is $\ket{I, PC}_{AB}$, and the Pauli error $P$ is efficiently calculable.
	\end{lemma}
	
	\begin{proof}
		We divide the $2mn$ qubits into $m$ blocks, each containing $2n$ qubits, and prepare EPR pairs within each block. For each $i$-th block, we apply the circuits $C_{2i-1}^T \otimes C_{2i}$ to the respective sides of the EPR pairs. According to Eq.~\eqref{eq:move_unitary_EPR}, the state $\ket{\phi^n}$ in the $i$-{th} block evolves to $\ket{C_{2i-1}^T, C_{2i}} = \ket{I, C_{2i} C_{2i-1}}$.
		
		We continue this process, merging the outputs of consecutive blocks using Bell state measurements as described in Lemma \ref{lem:teleport_clifford}. Specifically, for the states $\ket{I, C_2 C_1}$ and $\ket{I, C_4 C_3}$, we apply the Bell state measurement to obtain $\ket{I, P_2 C_4 C_3 C_2 C_1}$, where $P_2$ is a Pauli error. Repeating this for all blocks, we obtain:
		\begin{equation}
			\ket{I, P_m C_{2m} C_{2m-1} P_{m-1} C_{2m-2} C_{2m-3} \cdots C_1} = \ket{I, P C_{2m} C_{2m-1} C_{2m-2} C_{2m-3} \cdots C_1} = \ket{I, PC}.
		\end{equation}
		
		The Bell state measurements are performed simultaneously in a single layer, ensuring the overall circuit depth remains at $d + 2$, where the additional two layers account for the preparation of EPR pairs and the Bell state measurement. The Pauli error $P$ can be efficiently calculated in the Heisenberg picture, as described in \cite{gottesman1998heisenberg}.
	\end{proof}

	Now, we can prove Theorem \ref{thm:clifford_volume}, which shows a spacetime conversion for implementing Clifford circuits. 
	
	\begin{proof}[Proof of Theorem \ref{thm:clifford_volume}]
		Given $k$, define $d_2 = \left\lfloor \frac{t}{k} \right\rfloor + 2$. Decompose $C$ as $C = C_{k}C_{k-1}\cdots C_1$, where each $C_i$ represents a Clifford circuit of depth no greater than $d_2$. For the first claim, when $k = 2m$, Lemma \ref{lem:clifford} enables the preparation of $\ket{I, PC}_{AB}$ within a depth of $d_2 + 2$. A computational basis measurement on subsystem $A$ produces a result $\mathbf{a} \in \{0,1\}^n$, leading to
		\begin{equation}
			\ket{\varphi} = PC \ket{\mathbf{a}} = PC X^{\mathbf{a}} \ket{0}^{\otimes n} = PP' C \ket{0}^{\otimes n},
		\end{equation}
		where $P' = C X^{\mathbf{a}} C^{\dagger}$ is a Pauli string. Then,  the state $C\ket{0}^{\otimes n}$ can be obtained by applying $P'P$ to $\ket{\varphi}$. 
		
		When $k = 2m + 1$, denote $C'=C_{k-1}C_{k-2}\cdots C_1$. One can prepare $\ket{I,PC'}$
		and $C_k\ket{0}^{\otimes n}$ within a depth of $d_2 + 2$. Then, perform Bell state measurements on these two states and obtain measurement results $\ba$ and $\bb$. The post-measurement state is
		\begin{equation}
			\ket{\varphi} = C_kP'PC' \ket{0}^{\otimes n} = P''C \ket{0}^{\otimes n},
		\end{equation}
		where $P' = X^{\ba}Z^{\bb}$ and $P'' = C_kP'PC_k^{\dagger}$. Then, the state $C\ket{0}^{\otimes n}$ can be obtained by applying $P''$ to $\ket{\varphi}$. This method directly applies to the second claim by applying $C_k$ to $\ket{\phi}$ instead of $\ket{0}^{\otimes n}$.
		
		The last Bell state measurement and Pauli error correction are executed simultaneously with previous measurements in the final layer, ensuring the total circuit depth is $d = d_2 + 2 = \left\lfloor \frac{t}{k} \right\rfloor + 4$.
	\end{proof}

	\section{Bounding embedded complexity by circuit volume}\label{app:bounding}
	
	Here we present the detailed proof of  Theorem 1. We first introduce the concepts of semi‑algebraic sets and accessible dimension, which serve as key tools in our analysis. Then, we proceed to prove the two parts of Theorem~1 separately.
	
	\subsection{Semi-algebraic sets}
	The notion of semi‑algebraic sets and their dimensions provides a powerful framework for characterizing the degrees of freedom in sets of quantum states and operations. This, in turn, can be used to derive lower bounds on circuit complexity. We begin with the formal definition:
	\begin{definition}[(Semi-)algebraic sets]
		A set $S \subseteq \bR^M$ is called a semi‑algebraic set if there exist sets of polynomial functions $\{f_j\}$ and $\{g_k\}$ such that
		\begin{equation}
			S = \left\{x \in \bR^M: f_j(x)=0, g_k(x) \le 0 \quad\text{for all }j,k \right\}.
		\end{equation}
		Moreover, if $\{g_k\} =\emptyset$, then $S$ is called an algebraic set.
	\end{definition}
	
	A useful method to determine if a set is semi-algebraic is through the Tarski--Seidenberg theorem, which states that  polynomial functions map semi-algebraic sets to semi-algebraic sets. Here, we say that $F$ is a polynomial function if each entry of $F(x)$ is a polynomial of entries of $x$. 
	
	\begin{fact}[Tarski--Seidenberg theorem]\label{fact:TS_principle}
		Let $F: \bR^{M_1} \rightarrow \bR^{M_2}$ be a polynomial function. If $S \subseteq \bR^{M_1}$ is semi‑algebraic, then the image $F(S) \subseteq \bR^{M_2}$ is also semi‑algebraic.
	\end{fact}

	We now show that the set of post‑measurement states considered in Theorem 1 forms a semi‑algebraic set. To define the set of post-measurement states formally, consider the transformation from a local quantum circuit on an $m$-qubit system, constructed of $d$ layers of 2-qubit gates, to an $n$-qubit post-measurement state. This state results from measuring $m-n$ qubits and postselect the outcome $0^{m-n}$. The total number of gates in the circuit is 
	\begin{equation}
		V = \lfloor m / 2\rfloor d,
	\end{equation}
	and the circuit consists of gates $U_1, U_2, \ldots, U_V$. The mapping that takes $V$ 2-qubit gates as input and outputs the unnormalized post-measurement state can be written as:
	\begin{equation}\label{eq:post_select}
		\begin{split}
			G : SU(4)^{V} &\rightarrow \mathbb{R}^{M}, \\
			G(U_1,U_2,\cdots,U_V) = (\bra{0}^{\otimes (m-n)} &\otimes I_{2^n}) U_VU_{V-1}\cdots U_1 \ket{0}^{m}.
		\end{split}
	\end{equation}
	where $M = 2^{n+1}$ represents the degrees of freedom of unnormalized pure states. 
	Let $\cC$ denote the image of the map $G$, representing the set of unnormalized post‑measurement states. We now show that $\cC$ is a semi‑algebraic set. This follows directly from the Tarski–Seidenberg theorem.
	
	\begin{lemma}\label{lem:C_semi_algebraic}
		The set $\cC$ is semi-algebraic.
	\end{lemma}
	\begin{proof}
		The set $\mathcal{C}$ is the image of $SU(4)^V$ under the mapping $G$. The group $SU(4)$ consists of $4 \times 4$ unitary matrices with determinant one, which can be described by polynomial constraints:
		\begin{equation}
			UU^{\dagger} = I \quad \text{and} \quad \det(U) = 1. 
		\end{equation}
		Since these are polynomial equalities over $\mathbb{R}^{16}$ (identifying complex entries with pairs of real numbers), $SU(4)$ is an algebraic set. Consequently, $SU(4)^V$ is also algebraic.
		
		To invoke the Tarski–Seidenberg theorem (Fact~\ref{fact:TS_principle}), it suffices to show that $G$ is a polynomial map. Note that $U = U_V U_{V-1} \cdots U_1$ is a product of matrices from $SU(4)$, and thus each entry of $U$ is a polynomial in the entries of the $U_i$. The post‑measurement state $ (\bra{0}^{\otimes (m-n)} \otimes I_{2^n}) U\ket{0}^{\otimes m}$ is a subset of entries of $U$, and hence each of its entries is still a polynomial function of the entries of $U_1, \ldots, U_V$. Therefore, $G$ is a polynomial function, and the image $\cC = G(SU(4)^V)$ is semi‑algebraic.
	\end{proof}
	
	\subsection{Accessible dimension}
	
	We now show how to characterize the degrees of freedom in post‑measurement states by employing the concept of accessible dimension. Informally, the accessible dimension quantifies the number of independent directions in which the post‑measurement state $G(x)$ can be perturbed by infinitesimally perturbing the point $x = (U_1, U_2, \ldots, U_V) \in SU(4)^V$.
	
	To formalize this notion, we define the local perturbation map around a point $x = (U_1, U_2, \ldots, U_V)$ as follows:
	\begin{equation}
		\exp^V_x: (H_1, \ldots, H_V) \mapsto \left(\exp(iH_1) U_1, \ldots, \exp(iH_V) U_V\right),
	\end{equation}
	where each $H_i$ is a traceless Hermitian $4 \times 4$ matrix. They can be expanded in the Pauli basis as
	\begin{equation}
		H_i = \sum_{P \in \{I,X,Y,Z\}^{\otimes 2}, \, P \neq I} \lambda_{i,P} P.
	\end{equation}
	By definition, the map satisfies $\left.\exp^V_x\right|_0 = x$.

	We now compute the directional derivative of the composed map $G \circ \exp^V_x$ in the direction of each basis element $P$ of $H_i$. A small perturbation in $\lambda_{i,P}$ induces a first-order variation in the post‑measurement state given by
	\begin{equation}
		v_{x,i,P} \coloneqq \left.\frac{\partial}{\partial (i\lambda_{i,P})} G(\exp^V_x)\right|_{0}
		= (\bra{0}^{\otimes (m-n)} \otimes I_n)\, U_V \cdots U_{i+1} \, P \, U_i \cdots U_1 \ket{0}^{\otimes m}.
	\end{equation}
	The \emph{tangent space} $\mathcal{T}(x)$ at $x$ is defined to be the span of all such vectors:
	\begin{equation}
		\mathcal{T}(x) \coloneqq \mathrm{span}\left\{v_{x,i,P}\right\}_{i,P}.
	\end{equation}
	
	The accessible dimension is defined as the dimension of the tangent space $\mathcal{T}(x)$. It is implicitly dependent on the choice of the mapping $G$, which will be clear from context in the subsequent discussion.
	\begin{definition}[Accessible dimension]
		The accessible dimension of $x \in SU(4)^V$ is defined as $\dim \mathcal{T}(x)$.
	\end{definition}
	
	The following result states that the set of points in $SU(4)^V$ with maximal accessible dimension has unit measure.
	
	\begin{lemma}[Accessible dimension is maximal on a measure‑one subset]\label{lem:acc_dim_measure_1}
		Define $d_{\max} = \max_{x \in SU(4)^V} \dim \mathcal{T}(x)$. Then the set
		\begin{equation}
			R \coloneqq \left\{x \in SU(4)^V : \dim \mathcal{T}(x) = d_{\max} \right\}
		\end{equation}
		has measure one in $SU(4)^V$.
	\end{lemma}
	
	\begin{proof}
		Suppose there exists a point $x \in SU(4)^V$ such that $\dim \mathcal{T}(x) = d_{\max}$. For any $x' \in SU(4)^V$, the condition $\dim \mathcal{T}(x') < d_{\max}$ implies that all $d_{\max} \times d_{\max}$ minors of the matrix $(v_{x',i,P})_{i,P}$ vanish.
		Since each $v_{x',i,P}$ is a polynomial function of $x'$, each of these minors is a polynomial in the entries of $x'$. Hence, the set
		\begin{equation}
			R^c = \left\{x' \in SU(4)^V : \dim \mathcal{T}(x') < d_{\max} \right\}
		\end{equation}
		is an algebraic subset of $SU(4)^V$. Moreover, it is a proper subset of $SU(4)^V$, because it excludes at least one point $x$.  These two conditions together imply $R^c$ has measure zero, by the irreducibility property of the algebraic set $SU(4)^V$ (see Ref. \cite{Haferkamp_2022_linear} for a rigorous mathematical treatment).
		Consequently, $R$ has measure one.
	\end{proof}
	
	This property is important because it allows us to infer global dimensional properties from a \emph{single} local point. Furthermore, for each $x \in R$, there exists an open neighborhood $N_x \ni x$ such that $G(N_x)$ forms a manifold of dimension $d_{\max}$~\cite{Haferkamp_2022_linear}. Hence, the dimension of the semi‑algebraic set $\mathcal{C}$ is
	\begin{equation}
		\dim \cC  = d_{\max}.
	\end{equation}
	
	We will use a lower bound on the accessible dimension for local random circuits, corresponding to the case where the map $G$ acts with $m = n$.
	
	\begin{lemma}[Accessible dimension of local random circuits, adapted from Ref.~\cite{Haferkamp_2022_linear}] \label{lem:acc_dim_local_random_circuits}
		Consider the map from two‑qubit gates to a global unitary $U \in SU(2^n)$:
		\begin{equation}
			F_1: (U_1, U_2, \ldots, U_V) \mapsto \prod_{i=1}^V U_i,
		\end{equation}
		where $V = \lfloor n / 2 \rfloor d$, and each $U_i$ is a two‑qubit gate in a depth‑$d$ local random circuit on $n$ qubits. Then there exists a point $x \in SU(4)^V$ such that
		\begin{equation}
			\dim \mathcal{T}(x) \ge \min\left( \left\lfloor \frac{d}{n} \right\rfloor,\, 4^n \right).
		\end{equation}
		
		Furthermore, for the state‑generation map
		\begin{equation}
			F_2: (U_1, U_2, \ldots, U_V) \mapsto \left( \prod_{i=1}^V U_i \right) \ket{0}^{\otimes n},
		\end{equation}
		there exists a point $x \in SU(4)^V$ such that
		\begin{equation}
			\dim \mathcal{T}(x) \ge \min\left( \left\lfloor \frac{d}{n} \right\rfloor,\, 2^{n+1} - 1 \right).
		\end{equation}
	\end{lemma}
	
	\begin{proof}
		The proof follows from Ref.~\cite{Haferkamp_2022_linear}, reformulated in our notation. The key observation is that a depth‑$n$ local random circuit suffices to conjugate any Pauli operator to $Z_n$, the Pauli-$Z$ operator acting only on the final qubit. To see this, note that for any 2-qubit Pauli operators, there exists a two‑qubit Clifford gate mapping one to the other. Therefore, we can sequentially conjugate a general Pauli operator $P$ through the chain:
		\begin{equation}
			P \xrightarrow{C_{1,2}} P_{\ge2} \xrightarrow{C_{2,3}} P_{\ge3} \xrightarrow{} \cdots \xrightarrow{C_{n-1,n}} Z_n,
		\end{equation}
		where each $P_{\ge j}$ acts nontrivially only on qubits $\{j,j+1,\cdots,n\}$. This composition $C = C_{n-1,n} \cdots C_{1,2}$ conjugates $P$ to $Z_n = CPC^{\dagger}$ and can be implemented with a depth‑$n$ local random circuit.
		
		We now prove the first part concerning the mapping $F_1$. Let $x \in SU(4)^V$ be chosen such that the global unitary is
		\begin{equation}
			U = \prod_{i=1}^D C_i,
		\end{equation}
		where each $C_i$ is a Clifford unitary generated by a depth‑$n$ local random circuit, and $D = \left\lfloor \frac{d}{n} \right\rfloor$. For each $j = 1, \dots, D$, consider perturbing a two‑qubit gate $u_j$ between $C_j$ and $C_{j+1}$ in the direction of $Z_n$. This results in
		\begin{equation}
			v_{x,u_j,Z_n}  =  \left(\prod_{i=j+1}^D C_i\right)Z_n \left(\prod_{i=1}^j C_i\right) = U \left(\prod_{i=1}^j C_i\right)^{\dagger } Z_n \left(\prod_{i=1}^j C_i\right).
		\end{equation}
		Suppose $D \le 4^n$, and choose $D$ independent Pauli operators $\{P_1, \ldots, P_D\}$. By choosing each $C_j$ sequentially, we can ensure that
		\begin{equation}
			Z_n =  C_1 P_1 C_1^{\dagger} = C_2 C_1 P_2C_1^{\dagger} C_2^{\dagger} = \cdots = (\prod_{i=1}^DC_i) P_D (\prod_{i=1}^DC_i)^{\dagger}.
		\end{equation}
		Hence, we have
		\begin{equation}
			v_{x,u_j,Z_n} = U \left(\prod_{i=1}^j C_i\right)^{\dagger } Z_n \left(\prod_{i=1}^j C_i\right) = UP_j,
		\end{equation}
		and since the $P_j$ are linearly independent, the vectors $v_{x,u_j,Z_n}$ are also linearly independent. Hence,
		\begin{equation}
			\dim \mathcal{T}(x) \ge D.
		\end{equation}
		If $D > 4^n$, the dimension is upper‑bounded by the number of independent Pauli operators, so
		\begin{equation}
			\dim \mathcal{T}(x) \ge \min(D, 4^n),
		\end{equation}
		proving the first part.
		
		For the mapping $F_2$, a similar argument shows that
		\begin{equation}
			v_{x,u_j,Z_n}=  U P_j \ket{0}^{\otimes n}.
		\end{equation}
		By suitably choosing the $P_j$, we can ensure the vectors $P_j \ket{0}^{\otimes n}$ span different states in the computational basis (with additional phases) $\{(i)^\kappa\ket{x}\}_{\kappa \in \{0,1\}, x \in \{0,1\}^n}$, 
		which corresponds to applying $I$, $X$, $Y$, or $Z$ on the initial state. 
		Since a normalized $n$‑qubit quantum state has at most $2^{n+1} - 1$ real degrees of freedom, we obtain:
		\begin{equation}
			\dim \mathcal{T}(x) \ge \min(D, 2^{n+1} - 1),
		\end{equation}
		completing the proof.
	\end{proof}

	\subsection{Proof of Theorem 1: embedded complexity of projected states}
	We give here proof of Theorem 1, which states that the embedded complexity of a projected state can be lower-bounded by circuit volume. The main technique in proving this theorem is to analyze the accessible dimension of post-measurement states. This dimension intuitively represents the degrees of freedom within a semi-algebraic set. Firstly, we establish a lower bound on the accessible dimension of post-measurement states. Then, we demonstrate how to bound embedded complexity by accessible dimension. Combining these findings, we establish the theorem.
	
	The crucial observation is that random gate teleportation finds a point $x \in SU(4)^V$ for which the accessible dimension $\dim \mathcal{T}(x)$ is lower‑bounded by the circuit volume. Then, by Lemma~\ref{lem:acc_dim_measure_1}, we conclude that this lower bound holds on a measure‑one subset of $SU(4)^V$. In other words, the set of post‑measurement states is composed of high‑dimensional manifolds whose dimension is at least proportional to the circuit volume.
	
	\begin{lemma}[Lower bound on the accessible dimension]\label{lem:complexity_volume}
		For the map $G$, there exists a point $x \in SU(4)^V$ such that 
		\begin{equation}
			\dim \mathcal{T}(x) \ge \min(L, 2^{n+1}-1),
		\end{equation}
		where $L = \frac{md}{2n^2} - m(1 + \frac{1}{n} + \frac{1}{n^2}) -1$.
	\end{lemma}
	
	\begin{proof}
		We consider the 2-qubit gates acting on the first $kn$ qubits, where $k = \lfloor \frac{m}{n} \rfloor \ge \frac{m}{2n}$, and set other gates to identity. In Theorem 3, we demonstrate that a total circuit depth of $\lfloor t/k \rfloor + 4$ is sufficient to generate the state ensemble $\cS_{n,t}$ on the first $n$ qubits. In that proof, the number of layers of random circuits after Bell state preparation is
		\begin{equation}
			d_1 = \lfloor t/k \rfloor + 2,
		\end{equation}
		followed by Bell state measurement.
		
		To perform Bell state preparation and measurement in a one-dimensional local circuit, one initially prepares EPR pairs between qubits $1$ and $n+1$ using $n+1$ layers of gates. This involves swap operations to position qubits $1$ and $n+1$ adjacently, executing a 2-qubit gate, and restoring their original positions. By replicating this process for $n$ additional times, Bell states are prepared between qubits $i$ and $n+i$ for $1 \leq i \leq n$. Consequently, $n^2 + n$ layers are required to establish a maximally entangled state between qubits $1,2,\ldots,n$ and $n+1,n+2,\ldots,2n$. The preparation of the remaining EPR pairs can proceed in parallel. The same number of layers, $n^2 + n$, is adequate for Bell state measurement, leading to a total of $2(n^2 + n)$ layers for Bell state preparation and measurement.

		After using $2(n^2+n)$ layers for Bell state preparation and measurement, one utilizes $d_1 = d - 2(n^2+n)$ layers for the random circuits in the middle section. The state ensemble $\cS_{n,t}$ on the first $n$ qubits is contained in the set of post-measurement states, where 
		\begin{equation}
			t \ge k(d_1 - 2) \ge k[d - (2n^2+2n+2)].  
		\end{equation}
		
		It should be noted that the post-measurement state in the set $\cC$ remains unnormalized. Nevertheless, due to the properties of EPR pairs, the probability of obtaining $\ket{0}^{\otimes n}$ is consistently $c = 2^{-(k-1)n}$ (See Lemma \ref{lem:teleport_random_gate}). 
		
		The above argument shows the existence of a point $x \in SU(4)^V$ such that $G(x) \in c\mathcal{S}_{n,t}$. Moreover, perturbations to the two‑qubit unitaries situated between the Bell state preparations and Bell state measurements correspond to perturbations of $G(x)$ within the space $c\mathcal{S}_{n,t}$. 
		We can explicitly construct such a point $x$ by following the proof of Lemma~\ref{lem:acc_dim_local_random_circuits}, ensuring that
		\begin{equation}\label{eq:large_point_state}
			\dim \mathcal{T}(x) \ge \min(\lfloor t / n \rfloor ,2^{n+1}-1). 
		\end{equation} 
		Here, 
		\begin{equation} \label{eq:L_ineq}
			\begin{split}
				\lfloor t/n \rfloor    &\ge \left\lfloor \frac{k[d - (2n^2+2n+2)]}{n} \right \rfloor \\ 
				& \ge \frac{m[d-(2n^2+2n+2)]}{2n^2} - 1 \\
				& = \frac{md}{2n^2} - m(1 + \frac{1}{n} + \frac{1}{n^2}) - 1.
			\end{split}
		\end{equation}
	\end{proof}

	Moreover, we prove that the accessible dimension provides a lower bound for the embedded complexity.
	
	\begin{lemma}[Accessible dimension lower-bounds embedded complexity of states]\label{lem:acc_dim_bound_comp}
		Suppose the image of the mapping $G$ has a maximal accessible dimension $D$. Consider randomly selecting $V$ 2-qubit gates $U_1,U_2,\cdots,U_V$ from $SU(4)^V$. With unit probability, the post-measurement state $\ket{\phi} = G(U_1,U_2,\cdots,U_V)$ will satisfy  $\norm{\ket{\phi}} \neq 0$, and the embedded complexity of the normalized state $\ket{\psi} = \ket{\phi} / \norm{\ket{\phi}}$ is lower-bounded by:
		\begin{equation}
			C_{anc}(\ket{\psi}) \ge (D-1)/15. 
		\end{equation}
	\end{lemma}
	
	\begin{proof}
		Define $\cW(s)$ as the space of unnormalized states generated by applying $s$ 2-qubit gates on $2s$ qubits, with measurements performed in the middle of the circuit, and in the final layer on the  $2s-n$ ancillary qubits. We assume the measurement result is postselected by $\ketbra{0}$. Any other measurement outcome would yield the same set $\cW(s)$ due to the flexibility in choosing 2-qubit gates. Here, we only consider the qubit number up to $2s$ because there are at most $2s$ qubits that are non-trivial support of the quantum gates. 
		
		Then, we define another set
		\begin{equation}
			\cV(s) = \{c \ket{\psi} : c \in \mathbb{R}, c\ge 0, \ket{\psi} \otimes \ket{0}^{2s-n} \in \cW(s)\}.
		\end{equation}
		That is, we multiply the post-measurement state by a non-negative normalization number. By analyzing the free parameters, we have:
		\begin{equation}
			\dim \cV(s) \le 15s + 1.
		\end{equation}
		
		Choose $s_0$ to be the largest integer such that $15s_0+1 < D$, then we have 
		\begin{equation}
			\dim \cC > \dim \cV(s_0).
		\end{equation}
		
		We decompose $SU(4)^V$ into $R \cup R^c$, where $R$ is the set of regular points at which $G$ has a maximal accessible dimension, and the complement $R^c$ is the set of critical points. The preimage $G^{-1}(\cV(s_0))$ of $\cV(s_0)$ can then be expressed as:
		\begin{equation}
			(G^{-1}(\cV(s_0)) \cap R^c) \bigcup (G^{-1}(\cV(s_0)) \cap R).
		\end{equation}
		
		The first term has measure zero because $R^c$ is a measure-zero subset of $SU(4)^V$ by Lemma~\ref{lem:acc_dim_measure_1}. To show that the second term also has measure zero, note that for any point $x \in R$, there exists a neighborhood $N_x \subseteq R$ such that $G(N_x)$ is a manifold of dimension $D$. Since $\mathcal{V}(s_0)$ has strictly smaller dimension than $G(N_x)$, its preimage under $G$ within $N_x$ is a measure-zero subset (see Ref.~\cite{Haferkamp_2022_linear} for a rigorous justification).
		Given any $\varepsilon > 0$, we can find a compact subset $K \subseteq R$ such that the measure of ${R \mathbin{\backslash} K}$ is less than $\varepsilon$. By compactness, there exists a finite subcover $\{N_x\}_x$ of $K$, and thus the set $G^{-1}(\mathcal{V}(s_0)) \cap K$ is a measure-zero subset of $K$. Taking the limit $\varepsilon \rightarrow 0$, we conclude that $G^{-1}(\mathcal{V}(s_0))$ is a measure-zero subset of $SU(4)^V$.
		
		Consequently, randomly draw $V$ 2-qubit unitaries $U_1,U_2,\cdots, U_V$, with unit probability, the corresponding state $\ket{\phi} = G(U_1,U_2,\cdots)$ will not be in the set $\cV(s)$, thus the normalized state $\ket{\psi} = \ket{\phi} / \norm{\ket{\phi}}$ has an embedded complexity \begin{equation}
			C_{anc}(\ket{\psi}) \ge s_0 + 1 \ge (D-1)/15. 
		\end{equation}
	\end{proof}
	
	Combining Lemma \ref{lem:complexity_volume} and Lemma \ref{lem:acc_dim_bound_comp}, we establish a lower bound of embedded complexity by circuit volume. 
	
	\begin{proof}[Proof of Theorem 1: embedded complexity of projected states]
		Lemma \ref{lem:complexity_volume} shows that the accessible dimension of the image of $G$ satisfies
		\begin{equation}\label{eq:D_lower}
			D \ge \min\left(L, 2^{n+1}-1\right),
		\end{equation} 
		where $L = \frac{md}{2n^2} - m(1 + \frac{1}{n} + \frac{1}{n^2}) - 1$. Consider randomly drawing 2-qubit gates $U_1,U_2,\cdots,U_V$, denote the state $\ket{\phi} = U_VU_{V-1}\cdots U_1\ket{0}^{\otimes m}$. From Lemma \ref{lem:acc_dim_bound_comp}, we conclude that with unit probability, the probability of getting measurement result $0^{m-n}$ is non-zero, and the embedded complexity of the projected state $\ket{\psi}$ satisfies: 
		\begin{equation}\label{eq:anc_comp_lower}
			C_{anc}(\ket{\psi}) \ge (D-1)/15 = \min\left(L - 1, 2^{n+1}-2\right) / 15.
		\end{equation}
		
		This conclusion applies to an arbitrary measurement result, as applying random gates renders these measurement outcomes equivalent. Furthermore, since there are only a finite number of possible measurement results, with unit probability, all projected states have an embedded complexity satisfying Eq.~\eqref{eq:anc_comp_lower}.
		
		For $m \ge n \ge 4$, we have $L - 1 \ge \frac{md}{2n^2} - 2m$. If $d \ge 5n^2$ , we have  $L - 1\ge \frac{md}{10n^2} = \Omega(\frac{V}{n^2})$.
	\end{proof}
	
	\subsection{Embedded complexity of Kraus operators} \label{subsec:complexity_Kraus}
	Here, we prove the that the embedded complexity of the Kraus operators can be lower-bounded by circuit volume in local random circuits.  
	First, we define the embedded complexity of a Kraus operator as follows:
	\begin{definition} [Embedded complexity of Kraus operators]
		The embedded complexity $C_{anc}(K)$ of an $n$-qubit Kraus operator $K$ is defined as the minimal number of 2-qubit gates required to implement $K$ within an $n$-qubit subsystem embedded in a $m$-qubit measurement-assisted quantum circuits:
		\begin{equation}
			\begin{aligned}
				C_{anc}(K) \coloneqq & \min \{V: \exists m\ge n, K = (\bra{0}^{\otimes (m-n)} \otimes I_n) \\
				& \Pi_V U_V \Pi_{V-1}U_{V-1} \cdots \Pi_1 U_1 (\ket{0}^{\otimes (m-n)} \otimes I_n)\} 
			\end{aligned}
		\end{equation}
		The 2-qubit gates $U_i$ can be arbitrary unitaries in $SU(4)$ and may act on any pair of qubits. The projective operator $\Pi_i$ acts on the same pair of qubits as $U_i$,
		\begin{equation}
			\begin{split}
				&\Pi_i = P_{i,1} \otimes P_{i,2}, \\
				&P_{i,1}, P_{i,2} \in \{I, \ketbra{0}, \ketbra{1}\}.
			\end{split}
		\end{equation}
	\end{definition}
	
	We can establish the relation between circuit volume and embedded complexity of Kraus operators as follows:
	
	\begin{theorem}\label{thm:complexity_volume_Kraus}
		Given $m \ge n \ge 4$, consider a local random circuit $U \in \cU_{m,d}$ acting on the initial state $\ket{0}^{\otimes m}$. With unit probability, for any $a\in \{0,1\}^n$, the Kraus operator $K = (\bra{a} \otimes I_n) U (\ket{0}^{\otimes (m-n)} \otimes I_n)$ satisfies:
		\begin{equation}
			C_{anc}(K) \ge \min\left(\frac{md}{2n^2} - 2m, 4^n\right) / 15.
		\end{equation}
		For $d = \Omega(n^2)$, the bound can be made $C_{anc}(K) = \Omega(\min(\frac{V}{n^2}, 4^n))$, where $V = \lfloor m / 2 \rfloor d$ is the circuit volume.
	\end{theorem} 
	
	The proof follows the same structure as that of the proof of projected states. Firstly, similar to Eq.~\eqref{eq:post_select}, we define the set of Kraus operators, $\cK$, as the image of the map $H$, where
	\begin{equation}
		\begin{split}
			H : SU(4) &\rightarrow \mathbb{C}^{M \times M} \\
			H(U_1,U_2,\cdots,U_V) = (\bra{0}^{\otimes (m-n)} &\otimes I_{n}) U_VU_{V-1}\cdots U_1   (\ket{0}^{\otimes m-n} \otimes I_n).
		\end{split}
	\end{equation}
	Here, $M = 2^{n}$ represents the dimension of an $n$-qubit system, and $\mathbb{C}^{M \times M}$ denotes the space of all $n$-qubit Kraus operators. This set also forms a semi-algebraic set from the same argument as in Lemma \ref{lem:C_semi_algebraic}.
	
	\begin{lemma}
		The set $\cK$ is a semi-algebraic set.
	\end{lemma}

	Following the similar proof of Eq.~\eqref{eq:large_point_state}, we can show that there exists a point $x \in SU(4)^V$, such that
	\begin{equation}\label{eq:acc_dim_Kraus_operator}
		\dim \mathcal{T}(x) \ge \min(L, 4^n).
	\end{equation}
	where $L \ge \frac{md}{2n^2} - m(1 + \frac{1}{n} + \frac{1}{n^2}) - 1$. 
	The accessible dimension provides a lower bound on the embedded complexity, as stated in the following lemma.
	\begin{lemma}[Accessible dimension lower-bounds embedded complexity of Kraus operators]\label{lem:acc_dim_bound_comp_Kraus}
		Suppose the image of the mapping $H$ has a maximal accessible dimension $D$. Consider randomly selecting $V$ 2-qubit gates $U_1,U_2,\cdots,U_V$ from $SU(4)^V$. With unit probability, the Kraus operator $K = H(U_1,U_2,\cdots,U_V)$ will satisfy  
		\begin{equation}
			C_{anc}(K) \ge D/15. 
		\end{equation}
	\end{lemma}
	\begin{proof} 
		The proof follows that of Lemma \ref{lem:acc_dim_bound_comp} with one exception that the normalization factor $c$ present in the proof of Lemma \ref{lem:acc_dim_bound_comp} does not appear here. Consequently, the bound changes from $(D-1)/15$ to $D/15$. \end{proof}
	
	\begin{proof}[Proof of Theorem \ref{thm:complexity_volume_Kraus}: embedded complexity of Kraus operators]
		By combining Lemma \ref{lem:acc_dim_bound_comp_Kraus} with Eq.~\eqref{eq:acc_dim_Kraus_operator}, we conclude that for $a=0^{n-k}$, the Kraus  operator $K = (\bra{a} \otimes I_n) U (\ket{0}^{\otimes (m-n)} \otimes I_n)$ satisfies $C_{anc}(K) \ge \min(L,4^n) / 15$ with unit probability. Due to the property of Haar random, this conclusion holds for arbitrary $a \in \{0,1\}^{n-k}$. Since there are only finite many possible $a$, we conclude that with unit probability, for any $a$, the Kraus operator $K = (\bra{a} \otimes I_n) U (\ket{0}^{\otimes (m-n)} \otimes I_n)$ satisfies $C_{anc}(K) \ge \min(L,4^n) / 15$.
	\end{proof}
	
	\section{Approximate embedded complexity}\label{app:approx_embedded}
	In this section, we introduce and discuss \emph{approximate embedded complexity}, a robust notion of embedded complexity. We present two main results for approximate embedded complexity. First, we prove the doubly robust result that approximate state designs possess high approximate embedded complexity. Next, we demonstrate that the projected states obtained via random gate teleportation also exhibit high approximate embedded complexity. 
	
	\subsection{Definition}
	
	Because any state‑preparation routine on a quantum device is unavoidably noisy, and measurements are subject to statistical fluctuations dictated by Born’s rule, it is practically essential to take into account error tolerance, in which case the output state may only be an approximate version of the ideal  target state. For both practical reasons and mathematical convenience, it is commonly also assumed that every implementable two‑qubit gate is chosen from a finite universal gate set $\mathsf{S}$—for example, $\mathsf{S}=\{I, \mathrm{CNOT},H,T\}$.

	We formally define the approximate embedded complexity as follows:
	\begin{definition}[Approximate embedded complexity]
		Let $\mathsf{S}$ be a universal gate set. The $\varepsilon$‑approximate embedded complexity of an $n$‑qubit pure state $\ket{\psi}$ with respect to $\mathsf{S}$ is defined as
		\begin{equation}\label{eq:unnormalized_projected_state}
			\begin{split}
				C_{anc}^{(\mathsf{S},\varepsilon)}(\psi) = \min \{V: \exists &m\ge n,  \text{s.t. }\dtr (\ket{\psi} \otimes \ket{0}^{\otimes (m - n)}, \ket{\phi}) \le \varepsilon  \\
				& \text{where} \quad  \ket{\phi} = \frac{\ket{\widetilde{\phi}}}{\sqrt{\braket{\widetilde{\phi}}{\widetilde{\phi}}}}, \ket{\widetilde{\phi}} = \Pi_V U_V \Pi_{V-1}U_{V-1} \cdots \Pi_1 U_1 \ket{0}^{\otimes m} \neq 0\}.
			\end{split}
		\end{equation}
		The 2-qubit gates $U_i$ can be arbitrary unitaries in $\mathsf{S}$ and may act on any pair of qubits. The projective operator $\Pi_i$ acts on the same pair of qubits as $U_i$,
		\begin{equation}
			\begin{split}
				&\Pi_i = P_{i,1} \otimes P_{i,2}, \\
				&P_{i,1}, P_{i,2} \in \{I, \ketbra{0}, \ketbra{1}\}.
			\end{split}
		\end{equation}
	\end{definition}
	
	Here, $\dtr(\rho,\sigma) = \frac{1}{2} \norm{\rho-\sigma}_1$ denotes the trace distance. 
	Setting $\varepsilon=0$ and $\mathsf{S}=SU(4)$ recovers the original embedded complexity defined in the main text.
	By definition, no measurement‑assisted circuit that employs fewer than $C_{\mathrm{anc}}^{(\mathsf{S},\varepsilon)}(\psi)$ two‑qubit gates from $\mathsf{S}$ can prepare $\ket{\psi}$ within trace distance $\varepsilon$.
	Consequently, $C_{\mathrm{anc}}^{(\mathsf{S},\varepsilon)}(\psi)$ provides a lower bound on the gate cost of approximate state preparation.
	
	Even when a quantum device natively supports a continuous gate set such as $SU(4)$, the complexity $C_{\mathrm{anc}}^{(\mathsf{S},\varepsilon)}(\psi)$ for a fixed finite universal set $\mathsf{S}$ still offers a qualitative lower bound of the experimental overhead. 
	Any 2-qubit continuous gate can be approximated to within $\epsilon$ diamond distance by a sequence of gates from the set $\mathsf{S}$, with the sequence length growing only as $\mathrm{poly log}(\epsilon^{-1})$ by the Solovay–Kitaev theorem~\cite{Nielsen2016QCQI}. For a post-measurement state $\ket{\tilde{\phi}}$ with $V$ two-qubit gates, as given in Eq.~\eqref{eq:unnormalized_projected_state}, let $c = \sqrt{\braket{\tilde{\phi}}{\tilde{\phi}}}$. 
	We can find a state $\ket{\tilde{\varphi}}$ such that the trace distance $\dtr(\tilde{\phi}, \tilde{\varphi}) \le c\epsilon$ by approximating each gate with a diamond distance of $V^{-1}c\epsilon$. This leads to a total number of two-qubit gates given by $V' = V \, \mathrm{poly log}\left(V(c\epsilon)^{-1}\right)$.
	This approximation corresponds to approximating the normalized state $\ket{\phi}$ to a trace distance of $\epsilon$, i.e., $\dtr(\phi, \varphi) \le \epsilon$. Given that  practical experiments can only observe different experimental phenomena efficiently when the error $\epsilon = \Omega(1/\poly(V))$ and require the measurement probability $c^2$ to be polynomially small, we must have $c = \Omega(1/\poly(V))$. Therefore, it suffices to use gates from $\mathsf{S}$ with a gate count $V' = V \, \mathrm{poly log}(V)$, introducing at most a poly-logarithmic factor.
	In the remainder of our discussion, we restrict to discrete universal gate sets $\mathsf{S}$.

	\subsection{Result I: Approximate state designs}

	We now prove the fundamental property that states drawn from an approximate quantum state design typically possess high approximate embedded complexity. 
	The quantum state designs are defined to characterize the ``order'' of randomness in a state ensemble $\cS = \{p_i, \ket{\psi_i}\}$. A state ensemble $\cS$ is said to be a quantum state $t$-design if it reproduces the first $t$ moments of the Haar measure \cite{ambainis2007quantum}. Concretely, the $t$-th moments of state ensemble $\cS$ is calculated as 
	\begin{equation}
		M_{\cS}^{(t)} = \sum_i p_i \ketbra{\psi_i}^{\otimes t}.
	\end{equation}
	The Haar $t$-th moments $M_{\mathrm{Haar}}^{(t)}$ are simply defined with respect to the Haar measure, which is the unique uniform distribution over pure states in a Hilbert space. 
	The following definition gives a strong notion of approximate state design with multiplicative error $\epsilon$:
	\begin{definition}[Approximate state design]
		An ensemble $\cS$ is an $\epsilon$-approximate $t$-design if:
		\begin{equation}
			(1-\epsilon) M_{\mathrm{Haar}}^{(t)} \le M_{\cS}^{(t)}  \le (1+\epsilon) M_{\mathrm{Haar}}^{(t)},
		\end{equation}
		where $A \le B$ is an operator inequality meaning that $B-A$ is positive semidefinite.
	\end{definition}
	\noindent Approximate unitary $t$-designs can also be defined analogously \cite{Haferkamp2022randomquantum}.

	Our result is based on a counting argument: we bound the number of distinct states that a measurement‑assisted circuit using at most $G$ two‑qubit gates can prepare. Each application of a gate $U_i\in\mathsf{S}$ followed by a fixed two‑qubit measurement $\Pi_i$ increases the set of reachable states by at most a constant factor. We formalize the counting step in the lemma below.
	\begin{lemma}\label{lem:number_states_discrete}
		Measurement‑assisted quantum circuits composed of gates from a finite set $\mathsf{S}$ and containing at most $G$ two‑qubit gates can prepare at most $N$ distinct pure states, where
		\begin{equation}
			\log N =  \cO\bigl(G (\log G + \log \abs{\mathsf{S}})\bigr).
		\end{equation}
	\end{lemma}

	\begin{proof}
		Because only $G$ two‑qubit gates are applied, the circuit acts non‑trivially on at most $m=2G$ qubits. We may regard it as a depth‑$G$ circuit in which each layer contains a single two‑qubit gate followed by an optional measurement on those same qubits. In any layer, the gate can be placed on any of the $\binom{m}{2}$ pairs of qubits and can be chosen from $\mathsf{S}$. The following measurement on each qubit has three possibilities: no measurement, or a projective measurement with outcome $0$ or $1$. So each layer admits at most $\binom{m}{2}\,3^2\lvert\mathsf{S}\rvert$ distinct choices, and the total number of reachable states obeys
		\begin{equation}
			N \le \left[\binom{m}{2} \cdot 3^2 \abs{\mathsf{S}}\right]^G = \left[G (2G-1) \cdot 9 \mathsf{S}\right]^G
		\end{equation}
		Taking logarithms yields
		\begin{equation}
			\log N = \cO\bigl(G (\log G + \log \abs{\mathsf{S}})\bigr)
		\end{equation}
	\end{proof}
	
	On the other hand, concentration bounds for approximate state designs imply that a state sampled from an $\epsilon$‑approximate $k$‑design is, with high probability, far from every state in this finite set.
	\begin{lemma} \label{lem:design_far_single_state}
		Let $\cS$ be an $\epsilon$-approximate $k$-design over $n$ qubits, and let $\ket{\phi}$ be any pure $n$‑qubit state. We have  
		\begin{equation}
			\Pr_{\psi \sim \cS} [\dtr(\psi,\phi) \le \varepsilon] \le (1+\epsilon) D_k^{-1} (1-\varepsilon^2)^{-k/2},
		\end{equation}
		where $D_k \coloneqq \binom{2^n+k-1}{k}$ is the dimension of the symmetric subspace of $k$ copies of the $n$‑qubit Hilbert space.
	\end{lemma}
	
	\begin{proof}
		The proof leverages the explicit form of the Haar $k$-moment operator  $M_{\mathrm{Haar}}^{(t)}$ \cite{Mele2023IntroductionHaar}:
		\begin{equation}
			M^{(k)}_{\mathrm{Haar}} = \frac{1}{(2^n+k-1)\cdots(2^n+1)(2^n)} \sum_{\pi \in \mathcal{S}_k} \pi,
		\end{equation}
		where $\mathcal{S}_{k}$ is the symmetric group acting on $k$ copies of the $n$‑qubit Hilbert space $\mathcal{H}$.  Because $\tr(\pi \phi^{\otimes k}) = 1$ for any pure state $\phi$, 
		\begin{equation}
			\bE_{\psi \sim \mathrm{Haar}(n)}  \tr\left[ \psi^{\otimes k} \phi^{\otimes k}\right] = \tr\left[ M^{(k)}_{\mathrm{Haar}}  \phi^{\otimes k}\right] = \frac{k!}{(2^n+k-1)\cdots(2^n+1)(2^n)} = D_k^{-1}.
		\end{equation}
		Because $\mathcal{S}$ is an $\epsilon$‑approximate $k$‑design,
		\begin{equation}
			\begin{split}
				\bE_{\psi \sim \cS}\tr\left[ \psi^{\otimes k} \phi^{\otimes k}\right] &= \tr\left[ M^{(k)}_{\mathrm{\cS}}  \phi^{\otimes k}\right] \\
				&\le (1+\epsilon)  \tr\left[ M^{(k)}_{\mathrm{\mathrm{Haar}}}  \phi^{\otimes k}\right] \\
				&\le (1+\epsilon) D_k^{-1}.
			\end{split}
		\end{equation}
		For pure states $\psi$ and $\phi$, the trace distance satisfies $\dtr(\psi ,\phi) = \sqrt{1-\tr(\psi,\phi)^2}$. Therefore,
		\begin{equation}
			\begin{split}
				\Pr_{\psi \sim \cS} [\dtr(\psi,\phi) \le \varepsilon] &=  \Pr_{\psi \sim \cS} [\tr(\psi\phi) \ge \sqrt{1 - \varepsilon^2}] \\
				&= \Pr_{\psi \sim \cS} [\tr(\psi^{\otimes k}\phi^{\otimes k}) \ge (1 - \varepsilon^2)^{k/2}] \\
				&\le (1-\varepsilon^2)^{-k/2} \bE_{\psi \sim \cS}\tr\left[ \psi^{\otimes k} \phi^{\otimes k}\right]\\
				&\le (1+\epsilon) D_k^{-1} (1-\varepsilon^2)^{-k/2}.
			\end{split}
		\end{equation}
		where for the third line we utilized the Markov's inequality. This completes the proof.
	\end{proof}   
	
	Combining Lemmas \ref{lem:number_states_discrete} and \ref{lem:design_far_single_state} we obtain a lower bound on the approximate embedded complexity of states drawn from an approximate design.
	
	\begin{theorem}[Approximate embedded complexity in approximate state designs, formal version of Proposition 1 in the main text]\label{thm:approximate_embedded_complexity_design}
		
		Fix a universal gate set $\mathsf{S}$ and $\varepsilon\in(0,1)$.
		Let $\mathcal{S}$ be an $\epsilon$‑approximate $k$‑design on $n$ qubits. For a pure state $\ket{\psi}$ drawn from $\cS$, with probability at least $1-\delta$, we have
		\begin{equation}
			C_{anc}^{(\mathsf{S},\varepsilon)}(\psi) > G,
		\end{equation}
		provided that 
		\begin{equation}
			\varepsilon \le \sqrt{1-2^{-n/2}}, \quad 2^{n/2}\ge k \ge \Omega\left( n^{-1} G\left(\log G + \log \abs{\mathsf{S}}\right) + n^{-1} \log \delta^{-1}\right), \quad \epsilon = \cO(1).
		\end{equation}
	\end{theorem}
	\begin{proof}
		From Lemma \ref{lem:number_states_discrete}, the states generated by measurement-assisted quantum circuits that use at most $G$ gates from $\mathsf{S}$ can be listed as $\{\phi_i\}_{i=1}^N$, with
		\begin{equation}
			\log N = \cO(G (\log G + \log \abs{\mathsf{S}})).
		\end{equation}
		For each $\phi_i$, Lemma \ref{lem:design_far_single_state} gives 
		\begin{equation}
			\Pr_{\psi \sim \cS} [\dtr(\psi,\phi_i) \le \varepsilon] \le (1+\epsilon) D_k^{-1} (1-\varepsilon^2)^{-k/2}.
		\end{equation}
		A union bound over the $N$ states in $\{\phi_i\}_{i=1}^N$ gives
		\begin{equation}
			\Pr_{\psi \sim \cS} [\exists \,1 \le i \le N, \dtr(\psi,\phi_i) \le \varepsilon] \le N (1+\epsilon) D_k^{-1} (1-\varepsilon^2)^{-k/2}.
		\end{equation}
		Equivalently,
		\begin{equation}
			\Pr_{\psi \sim \cS} [C_{anc}^{(S,\varepsilon)}(\psi) \le G] \le N (1+\epsilon) D_k^{-1} (1-\varepsilon^2)^{-k/2}.
		\end{equation}
		To make the right‑hand side at most $\delta$, it suffices to require
		\begin{equation}
			\begin{split}
				\log \delta \ge \log N + \log (1+\epsilon) - \log (D_k) - \frac{k}{2} \log (1-\varepsilon^2). 
			\end{split}
		\end{equation}
		Setting $\epsilon=\cO(1)$ and using $D_k\ge(2^{n}/k)^{k}$ reduces the inequality to
		\begin{equation}
			kn - k \log \left( \frac{k}{\sqrt{1-\varepsilon^2}} \right) \ge \log N + \log \delta^{-1} + \cO(1) = \cO(G (\log G + \log \abs{\mathsf{S}}) + \log \delta^{-1}).
		\end{equation}
		Because $\varepsilon\le\sqrt{1-2^{-n/2}}, k\le 2^{n/2}$ implies $n-\log\bigl(k/\sqrt{1-\varepsilon^{2}}\bigr)\ge n/4$, we can choose 
		\begin{equation}
			k= \Omega\left( n^{-1} G\left(\log G + \log \abs{\mathsf{S}}\right) + n^{-1} \log \delta^{-1}\right).
		\end{equation}
		to satisfy the inequality, completing the proof.
	\end{proof}

	\subsection{Result II: Connecting approximate embedded complexity and circuit volume via random gate teleportation}
	
	Here, we discuss the connection between approximate state designs and local random circuits, and show that the random‑gate teleportation protocol connects circuit volume with both approximate state designs and approximate embedded complexity.
	
	Quantum state designs can be generated using polynomially many gates. For example, random Clifford states are known to form state 3-designs \cite{Zhu_2017}. Prior research indicates that local random circuits on $n$ qubits can form $\epsilon$-approximate $k$-designs in depth $t = \poly(n,k,\epsilon)$ \cite{Brand_o_2016, Haferkamp2022randomquantum}, and the dependence on $k$ was recently improved to linear scaling \cite{chen2024incompressibility}.
	\begin{fact}[Local random circuits are linear unitary $t$-design {\cite[Corollary 1.7]{chen2024incompressibility}}]\label{fact:g_scaling}
		\label{fact:linear_design}
		For $n \ge 2$ and $k \le \Theta(2^{2n/5})$, the local random circuits on $n$ qubits can form $\epsilon$-approximate unitary $k$-design in depth $t = g(n,k,\epsilon)$, where
		\begin{equation}
			g(n,k,\epsilon) = \cO\left((nk+ \log (\epsilon^{-1}))(\log k)^7 \right).
		\end{equation}
		Under the condition that $k \le \Theta(2^{2n/5})$, $g(n,k,\epsilon)$ can be made $O(n^8 k)$, in which the dependence on $\epsilon$ are hidden.
	\end{fact}

	Throughout this discussion we fix the gate set $\mathsf{S}$ to be a discrete universal gate set. Unless stated otherwise, the parameters $\delta$, $\epsilon$, and $\varepsilon$ are ignored to highlight the dependence on circuit volume. We also employ the symbols $\tilde{O}$ and $\tilde{\Omega}$ to indicate that logarithmic factors are suppressed.
	
	By combining Theorem \ref{thm:approximate_embedded_complexity_design} and Fact \ref{fact:linear_design} , the growth of approximate embedded complexity for random state ensemble $\cS_{n,t}$ can be established. 
	\begin{lemma} [Approximate embedded complexity of local-random-circuit states]
		\label{lem:approx_comp_linear}
		\label{lem:complexity_random_circuits}
		Given $\varepsilon = \cO(1), t \le \Theta(2^{2n/5})$, randomly drawing a state $\ket{\psi} \in \cS_{n,t}$, with high probability,
		\begin{equation}
			C^{(\mathsf{S},\varepsilon)}_{anc}(\psi) \ge \tilde{\Omega}\left(\frac{t}{n^7}\right).
		\end{equation}
	\end{lemma}
	\begin{proof}
		By Fact \ref{fact:linear_design}, the state ensemble $\cS_{n,t}$ forms an approximate $k$-design, where $k = \Omega(\frac{t}{n^8})$. Then, by Theorem  \ref{thm:approximate_embedded_complexity_design}, with high probability, the state $\ket{\psi} \in \cS_{n,t}$ has an approximate embedded complexity 
		\begin{equation}
			C_{\delta}(\ket{\psi}) = \tilde{\Omega} (nk) = \tilde{\Omega}\left(\frac{t}{n^7}\right).
		\end{equation}
		with high probability.
	\end{proof}
	
	The spacetime conversion for random circuits has useful implications for quantum state design and embedded complexity. Recently, considerable effort has been directed towards reducing the quantum circuit depth for generating quantum $k$-designs \cite{Nakata_2017, chen2024efficient1, metger2024simple, chen2024efficient2, chen2024incompressibility}. Our approach offers a simple yet powerful method to reduce the circuit depth by utilizing ancillary qubits. By combining Theorem 3 with Fact \ref{fact:linear_design}, we have:
	
	\begin{lemma}[State design via random gate teleportation]
		For $n > \cO(\log k)$, an $\epsilon$-approximate $k$-design on $n$ qubits can be generated with total qubit number $k'n$ and circuit depth $d$, provided that $d \ge g(n,k,\epsilon) / k'$.
	\end{lemma}
	
	By inserting the expression of $g(n,k,\epsilon)$ in Fact \ref{fact:g_scaling}, the order $k$ of state design scales as
	\begin{equation}
		k = \Omega\left(\frac{k'd}{n^8}\right) = \Omega\left(\frac{V}{n^9}\right),
	\end{equation}
	where $V = \Theta(k'nd)$ represents the circuit volume. Our result shows that, with the utilization of ancillary qubits and measurements, the order $k$ can scale linearly with the circuit volume $V$. This finding provides another operational meaning of circuit volume.

	Additionally, the approximate embedded complexity of the states generated by the random gate teleportation protocol can also be bounded by circuit volume.
	
	\begin{theorem}[Bounding approximate embedded complexity by circuit volume]
		Consider the $n$-qubit state ensemble $\mathcal{S}_{n,t}$ in Theorem 3 in the main text, where the total qubit number is $k'n$ and the depth is $d = \left\lfloor \frac{t}{k'} \right\rfloor + 4$. Randomly drawing a state $\ket{\psi} \in \mathcal{S}_{n,t}$, with high probability, the approximate embedded complexity satisfies
		\begin{equation}
			C^{(\mathsf{S},\varepsilon)}_{anc}(\psi) \ge  \tilde{\Omega}\left(\frac{V}{n^8}\right),
		\end{equation}
		provided that $d \le \Theta(2^{2n/5})$. Here, $V = \Theta(k'nd)$ represents the circuit volume.
	\end{theorem}
	\begin{proof}
		Theorem 3 shows that random circuits with depth $d$ on $k'n$ qubits is enough to generate $\cS_{n,t}$ on $n$ qubits, where $t \ge k'(d-4)$. By Lemma \ref{lem:approx_comp_linear}, these states has an approximate embedded complexity
		\begin{equation}
			C^{(\mathsf{S},\varepsilon)}_{anc}(\psi) = \tilde{\Omega} \left(\frac{t}{n^7}\right) = \tilde{\Omega}\left(\frac{k'(d-4)}{n^7}\right) = \tilde{\Omega} \left( \frac{V}{n^8}\right)
		\end{equation}
		with high probability.
	\end{proof}
	
	These findings, derived from a specific construction, underscore how the design order and the approximate embedded complexity of projected states in a subsystem scale with circuit volume. These findings further clarify the trade-off between space complexity and embedded complexity and provide evidence for the general behavior of approximate embedded complexity growth of projected states.
	
	\section{Approximate embedded complexity in time-independent Hamiltonian evolutions}\label{app:approx_embed_Hamiltonian}

	In this section we provide details for the embedded complexity result for Hamiltonian evolution outlined in the main text. Specifically, we show that projected states produced by the time‑independent evolution of a local Hamiltonian have approximate embedded complexity that is lower‑bounded by the circuit volume, defined as the product of the evolution time and the total system size.
	
	As a concrete example, we consider a two‑dimensional lattice with $m_r$ rows and $m_c$ columns, so the total number of qubits is $m = m_r m_c$. The Hamiltonian is 
	\begin{equation}
		H = \sum_i h_i X_i + \sum_{i,j} h_{i,j} X_iX_j,
	\end{equation}
	where the on‑site fields $h_i$ and interaction strengths $h_{i,j}$ will be specified later.
	Our argument proceeds in three steps: We first show that this Hamiltonian can prepare graph states. Then,  using measurement‑based preparation protocols for graph states \cite{Raussendorf2001_oneway,Turner_2016,Mezher2018GraphStates}, we demonstrate that graph states efficiently generate random projected states. Concretely, after a certain evolution time $\tau$, measuring the time‑evolved state $\exp(-iH\tau)\ket{0}^{\otimes m}$ in the computational basis produces random states that form approximate state designs. This indicates that the Hamiltonian evolution exhibits deep thermalization phenomenon \cite{Cotler_2023}.
	Finally, by the relationship between approximate state designs and approximate embedded complexity established in Section~\ref{app:approx_embedded}, we conclude that these projected states have approximate embedded complexity lower-bounded by the circuit volume.
	
	\subsection{Graph states from Hamiltonian evolutions}
	A graph state is defined with respect to a graph $G=(M,E)$, where $M$ is the set of $m$ vertices arranged on a two‑dimensional lattice and $E$ is the set of edges:
	\begin{equation}
		\ket{G} = \prod_{(i,j)\in E} \mathrm{C}Z_{i,j} \ket{+}^{\otimes m},
	\end{equation}
	with the $\mathrm{C}Z$ gates defined as $\mathrm{C}Z_{i,j} = \exp(-i\pi Z_iZ_j/2)$. 
	For each qubit $i$ we choose a measurement basis $\{\ket{+\alpha_i},\ket{-\alpha_i}\}$ specified by an angle $\alpha_i$, where
	\begin{equation}
		\ket{\pm \alpha_i} = \ket{0} \pm e^{i\alpha_i}\ket{1}.
	\end{equation}
	Since $\ket{+\alpha}=Z(\alpha)H\ket{0}$ and $\ket{-\alpha}=Z(\alpha)H\ket{1}$ with $Z(\alpha)=\exp(-i\alpha Z/2)$, the measurement can be implemented by first applying $HZ(-\alpha_i)$ to each qubit of the graph state and then measuring in the computational basis. The rotated graph state is
	\begin{equation}
		\ket{\widetilde{G}} = H^{\otimes m}\prod_{i} Z_i(-\alpha_i)\ket{G}. 
	\end{equation}
	Using the identity $HZH=X$, we obtain
	\begin{equation}
		\begin{split}
			\ket{\widetilde{G}} &= H^{\otimes m}\prod_{i} Z_i(-\alpha_i)\prod_{(i,j)\in E} \mathrm{C}Z_{i,j} \ket{+}^{\otimes m} \\
			&= H^{\otimes m}\prod_{i} \exp(i\alpha_i Z_i/2)\prod_{(i,j)\in E} \exp(-i\pi Z_iZ_j / 2) H^{\otimes m} \ket{0}^{\otimes m}\\
			&= \prod_i \exp(i\alpha_i X_i/2)\prod_{(i,j)\in E} \exp(-i\pi X_iX_j / 2) \ket{0}^{\otimes m}\\
			&= \exp\left[-i \left(-\sum_i \frac{\alpha_i}{2} X_i + \sum_{(i,j) \in E} \frac{\pi}{2} X_iX_j \right)\right] \ket{0}^{\otimes m}.
		\end{split}
	\end{equation}
	Assume the coupling strength is normalized by a constant $J_0$. By setting
	\begin{equation}\label{eq:interaction_strength}
		\begin{split}
			h_{i,j} &= J_0 \quad \text{for \,} (i,j) \in E, \\
			h_i &= -\frac{\alpha_i J_0}{\pi}, \quad \tau = \frac{\pi}{2J_0},
		\end{split}
	\end{equation}
	we can prepare $\ket{\widetilde{G}}=\exp(-iH\tau)\ket{0}^{\otimes m}$ via the Hamiltonian evolution generated by $H=\sum_i h_i X_i+\sum_{(i,j)\in E} h_{i,j} X_i X_j.$
	
	\subsection{Random projected states from graph states}
	
	We now show how to choose the edge set $E$ and measurement angles $\{\alpha_i\}$ so that the projected states form an ensemble of random states, with proof adapted from Ref.~\cite{Mezher2018GraphStates}.
	First, we can specify input and output vertices for a graph state. 
	Concretely, we fix a subset $I\subset V$ and replace the qubits on $I$ by an input state $\ket{\psi}$:
	\begin{equation}
		\ket{G(\psi)} = \prod_{(i,j)\in E} \mathrm{C}Z_{i,j} \left(\ket{+}^{\otimes (M\backslash I)} \otimes \ket{\psi}^I\right)
	\end{equation}
	Then, we measure a subset of vertices and obtain a projected state on the complementary set $O\subset V$.
	
	Our construction is based on the graph‑state gadget $\mathfrak{B}$ shown in Fig.~\ref{fig:GraphStates}(a). 
	The gadget has 2 rows and 13 columns, thus $26$ vertices in total.  
	The input subset $I$ is chosen to be the first two qubits of the gadget, while the output subset $O$ corresponds to the last two qubits.
	It takes a two‑qubit input state $\ket{\psi}$, measures every qubit except the last two in the basis $\{\ket{\pm\alpha_i}\}$ (the angles $\alpha_i$ are indicated in the figure), and outputs
	\begin{equation}
		\ket{\psi} \xrightarrow{\mathfrak{B}} U_{1,2} \ket{\psi}
	\end{equation}
	where $U_{1,2}$ is drawn uniformly from a universal two‑qubit gate set $\mathsf{B}$ \cite{Mezher2018GraphStates}. 
	
	\begin{figure}[!htbp]
		\includegraphics[width=0.7\textwidth]{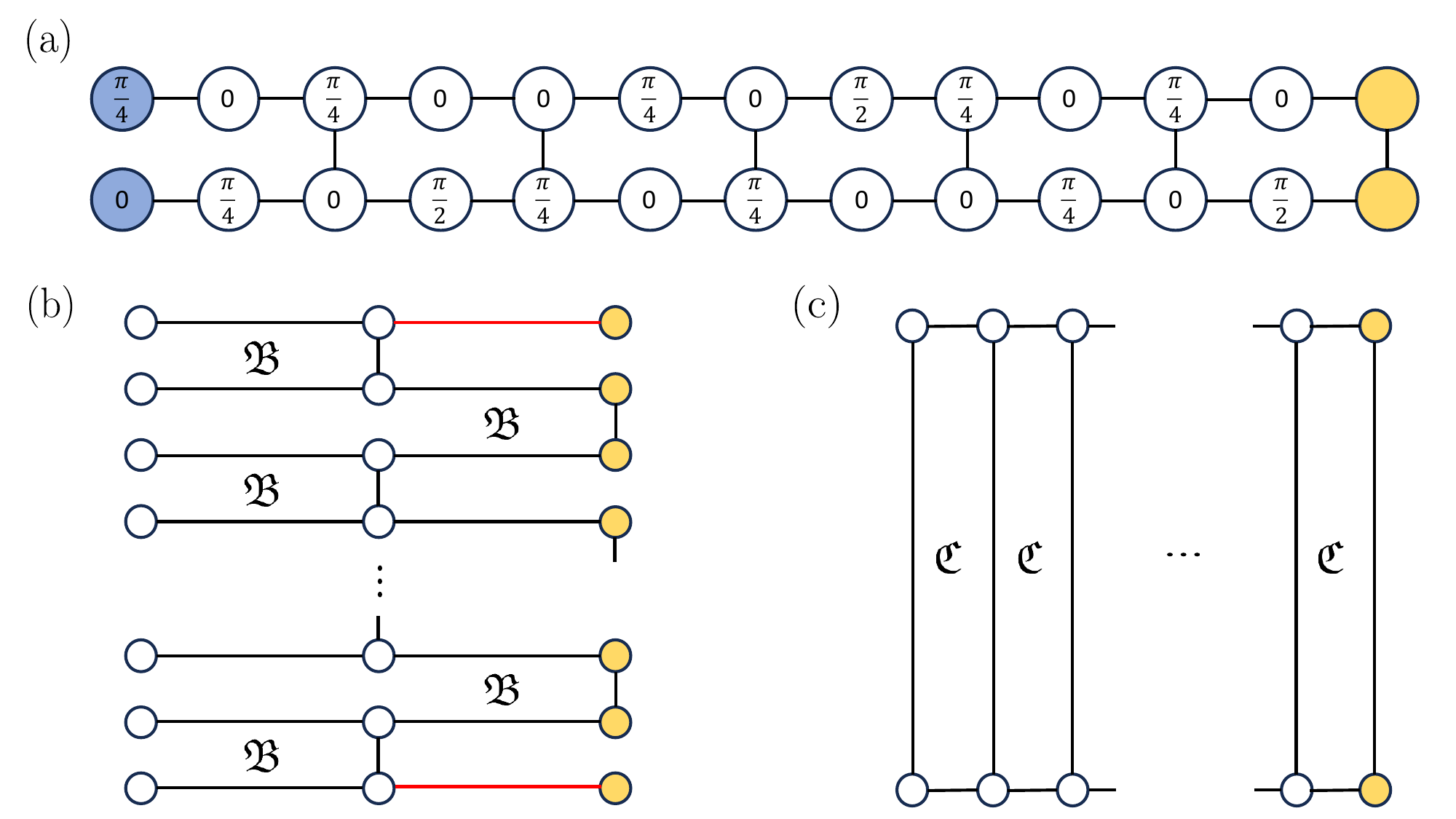}
		\caption{Construction of graph states. 
			(a) The gadget $\mathfrak{B}$ consists of $26$ vertices arranged in $2$ rows and $13$ columns. The first two qubits serve as inputs and the last two as outputs. The angles shown on each vertex indicate the measurement basis for that qubit.
			(b) The gadget $\mathfrak{C}$ is composed of two layers of gadget $\mathfrak{B}$ in a staggered arrangement, realizing two layers of local random circuits.
			(c) The whole graph state is built by concatenating multiple layers of $\mathfrak{C}$, with the projected state on the final column corresponding to the output of local random circuits applied to the initial state $\ket{+}^{\otimes m_r}$.} 
		\label{fig:GraphStates}
	\end{figure}
	
	By stacking copies of $\mathfrak{B}$ we build a larger gadget $\mathfrak{C}$ that applies two layers of local random circuits to an $m_r$‑qubit input state $\ket{\psi}$, as shown Fig.~\ref{fig:GraphStates}(b).
	We assume that $m_r$ is even, while the extension to odd values of $m_r$ is straightforward.
	A gadget $\mathfrak{C}$ is composed of two columes of gadgets $\mathfrak{B}$:
	\begin{enumerate}
		\item First column: $\mathfrak{B}$ act on pairs $(1,2),(3,4),\dots,(m_r-1,m_r)$, implementing $U^{(1)}= \bigotimes_{i}U_{2i-1,\,2i}^{(1)}$.
		\item Second column: $\mathfrak{B}$ act on pairs $(2,3),(4,5),\dots,(m_r-2,m_r-1)$, implementing $U^{(2)}=U_1^{(2)}\left(\bigotimes_{i}U_{2i,\,2i+1}^{(2)}\right)U_{m_r}^{(2)}$.
	\end{enumerate}
	Moreover, the twelve additional qubits in the first and last rows are measured with $\alpha=0$, effectively applying $\prod_{i=1}^{12} HZ^{o_i}$ to qubits $1$ and $m_r$, where $o_i\in\{0,1\}$ is the measurement outcome. This is equivalent to a single‑qubit unitary $U_1^{(2)}, U_{m_r}^{(2)}$ drawn uniformly from the set $\mathsf{A}=\{I,X,Y,Z\}$. In summary, the gadget $\mathfrak{C}$ performs
	\begin{equation}
		\ket{\psi}\xrightarrow{\mathfrak{C}}U^{(2)}U^{(1)}\ket{\psi},
	\end{equation}
	where every two‑qubit gate is drawn from $\mathsf{B}$ and all single‑qubit gates are drawn from $\mathsf{A}$.

	Finally, we concatenate $t$ copies of the gadget $\mathfrak{C}$, as illustrated in Fig.~\ref{fig:GraphStates}(c). The final output state on the last column becomes
	\begin{equation}
		\ket{\psi} \xrightarrow{\mathfrak{C} \times t} \prod_{i=1}^{2t} U^{(i)} \ket{\psi} 
	\end{equation}
	meaning that a depth‑$2t$ local random circuit has been applied to the $m_r$‑qubit input state $\ket{\psi}$. This construction corresponds to a measurement on a graph state defined over a two‑dimensional lattice of $m_r \times m_c$ qubits, with $m_c = 24t + 1$, and initial input state $\ket{+}^{\otimes m_r}$. The resulting ensemble of projected states is given by
	\begin{equation}
		\begin{split}
			\tilde{\mathcal{S}}_{m_r, 2t}
			\coloneqq \biggl\{
			\left(\prod_{i=1}^{2t} U^{(i)}\right)\ket{+}^{\otimes m_r} :
			U^{(2i-1)} = \prod_{j} U^{(2i-1)}_{2j-1,2j},\,
			U^{(2i)} = U_1^{(2i)}\left(\prod_j U_{2j,2j+1}^{(2i)}\right)U_{m_r}^{(2i)},\\
			U_{j,k}^{(i)} \sim \mathsf{B},\,
			U_1^{(2i)}, U_{m_r}^{(2i)} \sim \mathsf{A}
			\biggr\},
		\end{split}
	\end{equation}
	
	The ensemble $\tilde{\mathcal{S}}_{m_r,2t}$ is similar to the ensemble $\mathcal{S}_{m_r,2t}$ defined in the main text, except that (1) the initial state is $\ket{+}^{\otimes m_r}$ rather than $\ket{0}^{\otimes m_r}$, (2) the two‑qubit gates are drawn from $\mathsf{B}$ rather than from $SU(4)$, and (3) extra single‑qubit gates from $\mathsf{A}$ are applied at the first and last qubits of even layers.
	We summarize this construction in the following lemma.
	\begin{lemma}[Projected states of Hamiltonian evolutions]\label{lem:proj_evolution}
		Consider a Hamiltonian $H$ defined on an $m_r \times m_c$ square lattice with
		\[
		H = \sum_i h_i X_i + \sum_{(i,j)\in E} h_{i,j} X_i X_j,
		\]
		where $m_c = 24t + 1$, the edge set $E$ is specified by Fig.~\ref{fig:GraphStates}, and the coefficients $h_i$, $h_{i,j}$, and evolution time $\tau$ are given in Eq.~\eqref{eq:interaction_strength}. Let
		\[
		\ket{\psi} = \exp(-iH\tau)\ket{0}^{\otimes m_r m_c}.
		\]
		Then, measuring $m_r \times (m_c - 1)$ qubits of $\ket{\psi}$ in the computational basis produces a projected state on the final $m_r$ qubits drawn from the ensemble $\tilde{\mathcal{S}}_{m_r,2t}$.
	\end{lemma}

	\subsection{Approximate embedded complexity in Hamiltonian evolutions}
	We now establish a lower bound on the embedded complexity of projected states obtained from the above Hamiltonian evolution. Our argument proceeds by showing that these projected states form approximate state designs.
	As stated in Fact~\ref{fact:g_scaling}, the ensemble $\mathcal{S}_{n,t}$ forms approximate state designs of high orders. Ref.~\cite{chen2024incompressibility} further extends this result to universal gate sets, incurring additional $\mathrm{polylog}(k)$ factors in the depth. Here we adapt this result to our setting.
	
	\begin{fact}[\cite{chen2024incompressibility}]\label{fact:universal_design}
		Let $n \ge 2$ and $k \le 2^{\cO(n)}$.  The state ensemble $\tilde{\cS}_{n,t}$ forms an $\epsilon$‑approximate state $k$‑design in depth $t = g(n,k,\epsilon)$,
		where
		\begin{equation}
			g(n,k,\epsilon) = \cO\left([nk+ \log (\epsilon^{-1})]\mathrm{polylog}(k)\right).
		\end{equation}
		Under the condition that $k \le 2^{\cO(n)}$, $g(n,k,\epsilon)$ can be made $\poly(n) k$, in which the dependence on $\epsilon$ are hidden.
	\end{fact}
	
	Hence, the evolved state $\exp(-iH\tau) \ket{0}^{\otimes m}$ exhibit deep thermalizaton phenomenon \cite{Cotler_2023}. Based on this, we can establish the approximate embedded complexity of the projected states generated by the time-independent Hamiltonian evolution.
	
	\begin{theorem}[Approximate embedded complexity in time-independent Hamiltonian evolution]
		Let $\mathsf{S}$ be a finite two‑qubit universal gate set and consider an $m_r \times m_c$ lattice ($m=m_r m_c$ qubits) with Hamiltonian $H$ and evolution time $\tau$ as described in Lemma~\ref{lem:proj_evolution}. 
		After measuring $m_r(m_c - 1)$ qubits of the the evolved state $
		e^{-iH\tau}\ket{0}^{\otimes m}$ in the computational basis, the projected state $\ket{\psi}$ on the $m_r$ qubits in the last column satisfies, with probability at least $1-\delta$, 
		\begin{equation}
			C_{anc}^{(\mathsf{S},\varepsilon)}(\psi) = \tilde{\Omega}\left( \min\Bigl\{\frac{V}{ p_1(m_r)}, m_r2^{m_r/2}\Bigr\}\right).
		\end{equation}
		where $V = m\tau$ is the circuit volume and $p_1$ is a polynomial function, provided that
		\begin{equation}
			\varepsilon < \sqrt{1-2^{n/2}}, \quad  m_c = \Omega\left( m_r^{-1}p_1(m_r) \log(\delta^{-1})\right).
		\end{equation}
	\end{theorem}
	
	\begin{proof}
		By Lemma~\ref{lem:proj_evolution} and Fact \ref{fact:universal_design}, the projected ensemble is $\tilde{\cS}_{m_r,2t}$, which forms an $\epsilon$‑approximate $k$‑design with $\epsilon = \cO(1)$ and $k = t/p_1(m_r)$.  Theorem~\ref{thm:approximate_embedded_complexity_design} then implies that, when
		\begin{equation}
			2^{m_r/2} \ge k \ge \Omega\Bigl(m_r^{-1}G\bigl(\log G + \log |\mathsf{S}|\bigr) + m_r^{-1}\log(\delta^{-1})\Bigr),
		\end{equation}
		we have $C_{anc}^{(\mathsf{S},\varepsilon)} > G$ with probability at least $1-\delta$.  Choosing $m_c = \Omega\left( m_r^{-1}p_1(m_r) \log(\delta^{-1})\right)$ ensures $m_r^{-1}\log(\delta^{-1}) = \cO(k)$, yielding 
		\begin{equation}
			C_{anc}^{(\mathsf{S},\varepsilon)}(\psi)
			= 
			\tilde{\Omega}\left(m_r \min\{k,2^{m_r/2}\}\right)
			= 
			\tilde{\Omega}\!\left( \min\Bigl\{\frac{m_rt}{p_1(m_r)}, m_r2^{m_r/2}\Bigr\}\right)
			= 
			\tilde{\Omega}\!\left(\min\Bigl\{\frac{m_r m_c \tau}{p_1(m_r)},m_r2^{m_r/2}\Bigr\}\right),
		\end{equation}
		where we use $m_c = 24t+1$ and $\tau = \cO(1)$.
	\end{proof}
	
	\section{Classical hardness of sampling from random-gate-teleportation circuits}\label{app:hardness_sampling}
	
	In this section, we provide complexity-theoretic evidence that sampling from random-gate-teleportation (RGT) circuits is as hard as sampling from standard random circuits of comparable circuit volume acting on a subsystem, {solidifying our key message that the circuit volume establishes a fundamental spacetime characterization of the complexity of quantum systems}. Specifically, we consider RGT circuits acting on $m = (2k+1)n$ qubits, where $k \in \mathbb{N}^+$ and the first $n$ qubits constitute the subsystem of interest for random circuit sampling. The complexity analysis presented here extends straightforwardly to the case $m = 2kn$ as well.
	
	Recall that a RGT circuit first prepares $k$ Choi states
	\begin{equation}
		\ket{U_1^{T},U_2}_{A_1A_2},\,
		\ket{U_3^{T},U_4}_{A_3A_4},\,
		\ldots,\,
		\ket{U_{2k-1}^{T},U_{2k}}_{A_{2k-1}A_{2k}}
	\end{equation}
	together with the state $U_0\ket{0^{\otimes n}}_{A_0}$, where each $A_i$ is an
	$n$-qubit subsystem and every $U_i$ is a depth-$d$ local
	random circuit for $0\le i\le 2k$.  
	Bell-state measurements are then performed on the pairs
	$A_0A_1,A_2A_3,\cdots,A_{2k-2}A_{2k-1}$, yielding outcomes
	$\mathbf{a}_0\mathbf{a}_1\cdots\mathbf{a}_{2k-2}\mathbf{a}_{2k-1}$,
	followed by a computational-basis measurement on $A_{2k}$ with outcome
	$\mathbf{a}_{2k}$, where each $\mathbf{a}_i\in\{0,1\}^n$.
	The joint outcome probability is  
	\begin{equation}\label{eq:relation_probability}
		p(\ba_0\ba_1\cdots\ba_{2k}) = 2^{-2k} \abs{\bra{\ba_{2k}} U\ket{0} }^2, \quad U \coloneqq U_{2k}U_{2k-1}X^{\ba_{2k-1}}Z^{\ba_{2k-2}}U_{2k-2}U_{2k-3}\cdots U_1X^{\ba_1}Z^{\ba_0}U_0.
	\end{equation}

	We now show that sampling from the distribution $p$ is classically hard, based on the same complexity-theoretic assumptions underpinning the hardness of random circuit sampling. The key intuition is that, conditioned on the outcomes
	$\mathbf{a}_0,\mathbf{a}_1,\ldots,\mathbf{a}_{2k-1}$,
	the unitary $U$ is a random depth-$t$ circuit on $n$ qubits, where $t\coloneq(2k+1)d$.
	Therefore, sampling from $p$ has comparable hardness of as sampling from random quantum circuits of depth $t$ on $n$ qubits. 
	We formalize this intuition in the following, using a proof strategy that parallels the standard approach for establishing the classical hardness of random circuit sampling~\cite{Hangleiter2023AdvantageSampling}.
	
	Here, we denote $\tilde{\mathcal{C}}_d$ as the family of RGT circuits on $m$ qubits, where each block $U_i$ has depth $d$, and let $\mathcal{C}_t$ denote the family of depth-$t$ circuits on $n$ qubits. 
	
	\subsection{Worst-case hardness with constant multiplicative error}
	
	We show the hardness result on classically sampling from the output distribution of circuits in $\tilde{\mathcal{C}}_d$ to within a constant multiplicative error.  
	Our argument relies on the following complexity-theoretic result.
	
	\begin{fact}[Multiplicative-error sampling hardness {\cite[Sec.\,IV C]{Hangleiter2023AdvantageSampling}}]\label{fact:multiplicative_error}
		Let $\mathcal{C}$ be a family of quantum circuits for which approximating the output probability $|\langle \mathbf{0}|C|\mathbf{0}\rangle|^{2}$ to multiplicative error $c=\mathcal{O}(1)$ for some $C\in\mathcal{C}$ is $\mathsf{GapP}$-hard.  
		Suppose there existed a classical polynomial-time algorithm that, for any $C\in\mathcal{C}$, outputs samples from a distribution $q(x)$ satisfying
		\begin{equation}
			c_1\,p_C(x)\leq(x)\le\frac{p_C(x)}{c_1},
			\qquad
			p_C(x)\coloneqq|\langle x|C|\mathbf{0}\rangle|^{2},
		\end{equation}
		for some constant $c_1 < c$.  
		Then the polynomial hierarchy would collapse to its third level $\Sigma_{3}$.
	\end{fact}
	
	The requirement that approximating $p_C(0)$ to $\mathcal{O}(1)$ multiplicative error be $\mathsf{GapP}$-hard is known to hold for many circuit families, including instantaneous-quantum-polynomial circuits which are non-universal \cite{Bremner2010IQP_Circuits}, and for the depth-$t$ circuit family $\mathcal{C}_t$ whenever $t$ exceeds a certain threshold \cite{Terhal2004ConstantDepth}.
	We now show that the same hardness holds for $\tilde{\mathcal{C}}_d$ whenever it holds for $\mathcal{C}_{t}$.
	
	\begin{theorem}
		If approximating the output probability $p_C(0)$ of some circuits $C\in\mathcal{C}_{t}$ to $\mathcal{O}(1)$ multiplicative error is $\mathsf{GapP}$-hard, then the same is true for some circuits in $\tilde{\mathcal{C}}_{d}$.  
		Consequently, unless the polynomial hierarchy collapses to its third level $\Sigma_{3}$, no classical polynomial-time algorithm can sample from the output distribution $p_{C}(x)$ of $C\in\tilde{\mathcal{C}}_{d}$ within $\mathcal{O}(1)$ multiplicative error.
	\end{theorem}
	
	\begin{proof}
		From Eq.~\eqref{eq:relation_probability} we have
		\begin{equation}
			p_{C}(\mathbf{0})
			= 2^{-2k}p_U(\mathbf{0}),
		\end{equation}
		where $C\in\tilde{\mathcal{C}}_{d}$ and the corresponding $U\in\mathcal{C}_{t}$.  
		Suppose we could approximate $p_{C}(\mathbf{0})$ to multiplicative error $c=\mathcal{O}(1)$, i.e.,
		\begin{equation}
			c2^{-2k}p_U(\mathbf{0}) \le q \le \frac{2^{-2k}p_U(\mathbf{0})}{c}.
		\end{equation}
		Multiplying $q$ by $2^{2k}$ yields an approximation of $p_{U}(0)$ with the same multiplicative error $c$.  
		By assumption, producing such an approximation for some $U\in\mathcal{C}_{t}$ is $\mathsf{GapP}$-hard. 
		As a result, the approximation task is also $\mathsf{GapP}$-hard for $\tilde{\mathcal{C}}_{d}$, and the sampling hardness follows from Fact~\ref{fact:multiplicative_error}.
	\end{proof}

	\subsection{Average-case hardness with additive error}
	Even fault-tolerant quantum devices inevitably sample from a noisy distribution $p$ that differs from the ideal distribution $p_U$ by an additive error $\epsilon$ in total-variation distance,
	\begin{equation}
		\dTV{p-p_U} \le \epsilon.  
	\end{equation}
	As a result, practically, it is  more meaningful to express hardness results in this additive-error metric.  In practice one focuses on average-case hardness, since multiplicative-error guarantees do not straightforwardly imply additive-error bounds for a single instance of a circuit.
	
	\begin{fact}[Average-case hardness with additive error {\cite[Theorem 17]{Hangleiter2023AdvantageSampling}}]\label{fact:hardness_average_additive}
		Let $\mathcal C$ be a circuit family that satisfies  
		\begin{enumerate}
			\item[(i)] the hiding property, and  
			\item[(ii)] $\mathsf{GapP}$-hardness of approximating $p_C(\mathbf 0)$ on any $1-\delta$ fraction of circuits $C\in\mathcal C$ to accuracy
			\begin{equation}\label{eq:precision_hardness}
				\frac{1}{\poly(n)}p_C(\mathbf{0}) + \frac{2\epsilon}{2^n \delta}\left(1 + \frac{1}{\poly(n)}\right).
			\end{equation}
			where $\mathrm{poly}(n)$ is any polynomial.  
		\end{enumerate}
		Then, unless the polynomial hierarchy collapses, no classical polynomial-time algorithm can, with probability at least $1-\delta$ over a random $C\sim\mathcal C$, produce samples from $p$ satisfying $\dTV{p-p_C} \le \epsilon$.
	\end{fact}
	
	A circuit family $\mathcal{C}$ is said to possess the \emph{hiding property} if, for every circuit $C\in\mathcal{C}$ and every bit string $\mathbf{a}\in\{0,1\}^{n}$, one can efficiently construct a circuit $C'\in\mathcal{C}$ such that  
	\begin{equation}
		p_C(\ba) = p_{C'}(\mathbf{0})
	\end{equation}
	and, when $\mathbf{a}$ is drawn uniformly at random, the induced distribution of $C'$ matches that of an independently sampled circuit from $\mathcal{C}$:  
	\begin{equation}
		\Pr_{C'\sim \mathcal{C}}[C'] = \Pr_{C\sim \mathcal{C}, \ba \sim\{0,1\}^n} [C'].
	\end{equation}
	
	For the RGT circuit ensemble, Eq.~\eqref{eq:relation_probability} shows how to build such a circuit efficiently: replace each block $U_iU_{i-1}$ by $X^{\mathbf a_{i+1}}Z^{\mathbf a_i}U_iU_{i-1}$ and the final block $U_{2k}U_{2k-1}$ by $X^{\mathbf a_{2k}}U_{2k}U_{2k-1}$.  Because the Haar measure on $SU(4)$ is unitarily invariant, the resulting circuit $C'$ is distributed exactly according to $\tilde{C}_d$, so the RGT ensemble satisfies the hiding property.
	
	To satisfy condition~(ii), it is common in the literature to simplify the required precision to
	either an additive error of $\mathcal{O}(2^{-n})$ obtained via Markov’s inequality, or a constant relative error obtained via anticoncentration.  
	Following Ref.~\cite{Hangleiter2023AdvantageSampling}, we analyze these two scenarios separately and show that the average-case hardness of sampling from RGT circuits rests on the same complexity-theoretic assumptions as the standard hardness results for standard random-circuit sampling.
	
	\subsubsection{Hardness from approximating probability with $O(2^{-n})$ additive error}
	
	Condition~\eqref{eq:precision_hardness} in Fact~\ref{fact:hardness_average_additive} can be further simplified using Markov’s inequality.  
	Because the hiding property guarantees that the average of $p_C(\mathbf 0)$ over circuits $C\in\mathcal{C}$ is $2^{-n}$, we have
	\begin{equation}
		\Pr[p_C(\mathbf{0}) \ge \frac{1}{2^n\alpha}] \le \alpha
	\end{equation}
	for any $\alpha\in(0,1)$.  
	Therefore, with probability at least $(1-\alpha)$ over the choice of $C$, the quantity $p_C(\mathbf{0})$ is at most $2^{-n}/\alpha$.  
	On this fraction of instances, the additive term $\cO(2^{-n})$ in~\eqref{eq:precision_hardness} dominates.
	Consequently, if approximating $p_C(\mathbf 0)$ to additive error $\mathcal{O}(2^{-n})$ is $\mathsf{GapP}$-hard on \emph{any} $(1-\delta)(1-\alpha)$ fraction of the circuit ensemble $\mathcal{C}$, then condition~(ii) of Fact~\ref{fact:hardness_average_additive} is satisfied.  In other words, the average-case sampling hardness now reduces to the following requirement:
	\begin{itemize}
		\item It is $\mathsf{GapP}$-hard to approximate $p_C(\mathbf{0})$ within $\cO(2^{-n})$ additive error  on any $(1-\delta)(1-\alpha)$ fraction of the family $\mathcal{C}$.~\footnote{One must check that the failure probabilities $\delta$ and $\alpha$ can be treated independently; see Ref.~\cite{Hangleiter2023AdvantageSampling} for details.}
	\end{itemize}
	
	We now demonstrate that this requirement for $\mathcal{C}_t$ is already sufficient to establish the average-case sampling hardness of the RGT ensemble $\tilde{\mathcal{C}}_d$.
	\begin{theorem}
		Assume the $\mathsf{GapP}$-hardness of approximating the output probability $p_U(\mathbf{0})$ to additive accuracy $\mathcal{O}(2^{-n})$ on any $(1-\delta)(1-\alpha)$ fraction of depth-$t$ circuits $U\in\mathcal{C}_t$. 
		Then, unless the polynomial hierarchy collapses, no classical polynomial-time algorithm can, with probability at least $1-\delta$ over a random $C\sim \tilde{\mathcal C}_d$, produce samples from $p$ satisfying $\dTV{p-p_C} \le \epsilon$.
	\end{theorem}
	
	\begin{proof}
		We show that the assumed $\mathsf{GapP}$-hardness of probability
		approximation for the depth-$t$ family $\mathcal{C}_{t}$
		carries over to the RGT family $\tilde{\mathcal{C}}_{d}$.
		By Fact~\ref{fact:hardness_average_additive}, this implies the average-case sampling hardness of
		$\tilde{\mathcal{C}}_{d}$.
		
		Assume there exists a classical algorithm that, for a $(1-\delta)(1-\alpha)$ fraction of circuits $C \in \tilde{\mathcal{C}}_d$, outputs a value $q_0$ satisfying
		\begin{equation}
			\abs{q_0 - p_C(\mathbf{0})} \le \mathcal{O}(2^{-m}).
		\end{equation}
		By Eq.~\eqref{eq:relation_probability} we have $p_C(\mathbf{0}) = 2^{-2k} p_U(\mathbf{0})$ for some depth-$t$ circuit $U \in \mathcal{C}_t$. Multiplying the inequality by $2^{2k}$ gives
		\begin{equation}
			\abs{2^{2k} q_0 - p_U(\mathbf{0})} \le \mathcal{O}(2^{-n}).
		\end{equation}
		So the same algorithm approximates $p_U(\mathbf 0)$ to additive precision $\mathcal{O}(2^{-n})$ on a $(1-\delta)(1-\alpha)$ fraction of $\mathcal{C}_{t}$.  By assumption, this task is $\mathsf{GapP}$-hard.  Therefore the same task for circuits in $\tilde{\mathcal{C}}_{d}$ is also $\mathsf{GapP}$-hard.
	\end{proof}
	
	\subsubsection{Hardness from approximating probability with constant relative error}
	The hardness of sampling with additive error can also be reduced to the task of approximating a single output
	probability to constant relative error. That is, producing
	$q_0$ such that
	\begin{equation}
		\abs{q_0-p_C(\mathbf{0})} \le c p_C(\mathbf{0})
	\end{equation}
	for some constant $c>0$.
	To make this reduction we require an \emph{anticoncentration} property.
	
	\begin{definition}[Anticoncentration]
		A circuit family $\mathcal{C}$ anticoncentrates if for a constant $\alpha>0$, there exists a constant $\gamma(\alpha)>0$, independent of the system size $n$,  such that 
		\begin{equation}
			\Pr_{C \sim \mathcal{C}} [p_C(\mathbf{0}) \ge \frac{\alpha}{2^n}] \ge \gamma(\alpha).
		\end{equation}
	\end{definition}
	
	Given anticoncentration, condition~\eqref{eq:precision_hardness} in
	Fact~\ref{fact:hardness_average_additive} can be reduced to the following
	assumption:
	
	\begin{itemize}
		\item It is $\mathsf{GapP}$-hard to output $q_0$ satisfying
		\begin{equation}\label{eq:hardness_relative_error}
			\abs{q_0 - p_C(\mathbf{0})} \le \left(\frac{2 \epsilon}{\delta\alpha} + \frac{1}{\poly(n)}\right)p_C(\mathsf{\mathbf{0}}) \coloneqq cp_C(\mathbf{0})
		\end{equation}
		on any $\gamma(\alpha)(1-\delta)$ fraction of circuits
		$C\in\mathcal{C}$.
	\end{itemize}
	
	We have the following theorem:
	\begin{theorem}
		Assume the following two conditions hold:
		\begin{enumerate}
			\item The depth-$t$ circuit family $\mathcal{C}_t$ anticoncentrates.
			\item Approximating
			\(p_U(\mathbf 0)\) to constant relative error~$c$ on any
			$\gamma(\alpha)(1-\delta)$ fraction of circuits
			\(U\in\mathcal{C}_t\) is $\mathsf{GapP}$-hard.
		\end{enumerate}
		Then, unless the polynomial hierarchy collapses, no classical polynomial-time algorithm can, with probability at least $1-\delta$ over a random $C\sim \tilde{\mathcal C}_d$, produce samples from $p$ satisfying $\dTV{p-p_C} \le \epsilon$.
	\end{theorem}

	\begin{proof}
		By Eq.~\eqref{eq:relation_probability},  the anticoncentration of \(\tilde{\mathcal{C}}_d\) matches that of $\mathcal{C}_t$:
		\begin{equation}
			\Pr_{C\sim\tilde{\mathcal{C}}_d} [p_C(\mathbf{0}) \ge \frac{\alpha}{2^m}] = \Pr_{U\sim\mathcal{C}_{t}} [p_U(\mathbf{0}) \ge \frac{\alpha}{2^n}] = \gamma(\alpha).
		\end{equation}

		Moreover, any estimate $q_0$ satisfying the relative-error bound
		$\abs{q_0-p_C(\mathbf 0)}\le c\,p_C(\mathbf 0)$ yields
		$\abs{2^{2k}q_0-p_U(\mathbf 0)}\le c\,p_U(\mathbf 0)$,
		i.e., an estimate of \(p_U(\mathbf 0)\) with the same relative error.
		Therefore, if Condition~\eqref{eq:hardness_relative_error} holds for
		\(\mathcal{C}_t\), it also holds for \(\tilde{\mathcal{C}}_d\).
		Combining this with Fact~\ref{fact:hardness_average_additive} then implies the hardness of average-case sampling stated in the proposition.
	\end{proof}

	In summary, assuming the polynomial hierarchy does not collapse, we have shown:
	\begin{itemize}
		\item \textbf{Worst–case hardness.}  
		If approximating the probability $p_U(\mathbf 0)$ for depth–$t$ circuits $U\in\mathcal{C}_t$ to constant multiplicative error is worst-case $\mathsf{GapP}$-hard, then no efficient classical algorithm can sample from the RGT ensemble $\tilde{\mathcal{C}}_d$ within constant multiplicative error in the worst case.
		
		\item \textbf{Average–case hardness.}  
		If approximating $p_U(\mathbf 0)$ to additive precision $\cO(2^{-n})$ or to constant relative error is $\mathsf{GapP}$-hard on any constant fraction of $\mathcal{C}_t$, then no efficient classical algorithm can sample from $\tilde{\mathcal{C}}_d$ within additive error $\epsilon$ in the average case.
	\end{itemize}
	
	Existing proofs of random-circuit sampling hardness mainly reduce the sampling task to the same probability-approximation problems.  Consequently, the complexity-theoretic barriers for sampling RGT circuits $\tilde{\mathcal{C}}_d$ are on par with those for sampling from random circuits $\mathcal{C}_t$  with comparable circuit volume. See Ref.~\cite{Hangleiter2023AdvantageSampling} for a comprehensive discussion of random circuit sampling.
	
	\section{Ancilla-assisted shadow tomography}\label{app:shadow}
	We have shown that performing Bell state measurements can teleport random gates between subsystems in the main text. Building on this, we introduce an ancilla-assisted variant of the shadow tomography protocol, which aims to measure the properties of a state of interest $\rho$, accessible at the beginning of each experiment.

	Our protocol is inspired by the classical shadow protocol \cite{Huang_2020} that has drawn substantial recent interest. In the original classical shadow protocol, the states are measured in randomized bases, and this randomness is introduced by applying random unitaries $U$ to the input state $\rho \mapsto U \rho U^{\dagger}$. One needs to apply global random unitaries to estimate many important properties of the input state $\rho$, such as overlap fidelities with respect to many target states. This poses a major challenge for current quantum devices due to the highly sophisticated experimental controls required. Recent research has focused on easing the experimental requirement for classical shadow, such as developing shallow-depth classical shadow protocols \cite{bertoni2023shallow, hu2022classical}, or replacing the random unitaries with Hamiltonian evolutions \cite{Tran_2023, McGinley_2023, liu2024predicting}. 
	
	Here, we propose a protocol that avoids evolving the state $\rho$ under random unitaries or Hamiltonian dynamics. The key idea is to introduce the randomness from the ancillary system and Bell state measurements. We present the ancilla-assisted shadow tomography protocol in Box \ref{box:mbcs_protocol} and depict it in Fig.~\ref{fig:mbcs}.
	
	\begin{mybox}[label={box:mbcs_protocol}]{Ancilla-assisted shadow tomography}
		\textbf{Input:} 
		\begin{itemize}
			\item[] $N$ copies of an $n$-qubit state $\rho$.
		\end{itemize} 
		\textbf{Protocol:}
		\begin{enumerate}
			\item Select an $n$-qubit state ensemble $\mathcal{S}$.
			\item For each copy of input state $\rho$, randomly draw a state $\ket{\phi} \in \cS$.
			\item {Perform Bell state measurement between $\ket{\phi}$ and $\rho$, and record the measurement result.}
			\item  {Obtain $N$ data points by repeating Steps 2 to Step 3 on $N$ copies of $\rho$.}
			\item {Process the measurement results on classical computers to predict properties of the state $\rho$.}
		\end{enumerate}
	\end{mybox}
	
	\begin{figure}[!ht]
		\includegraphics[width=0.7\textwidth]{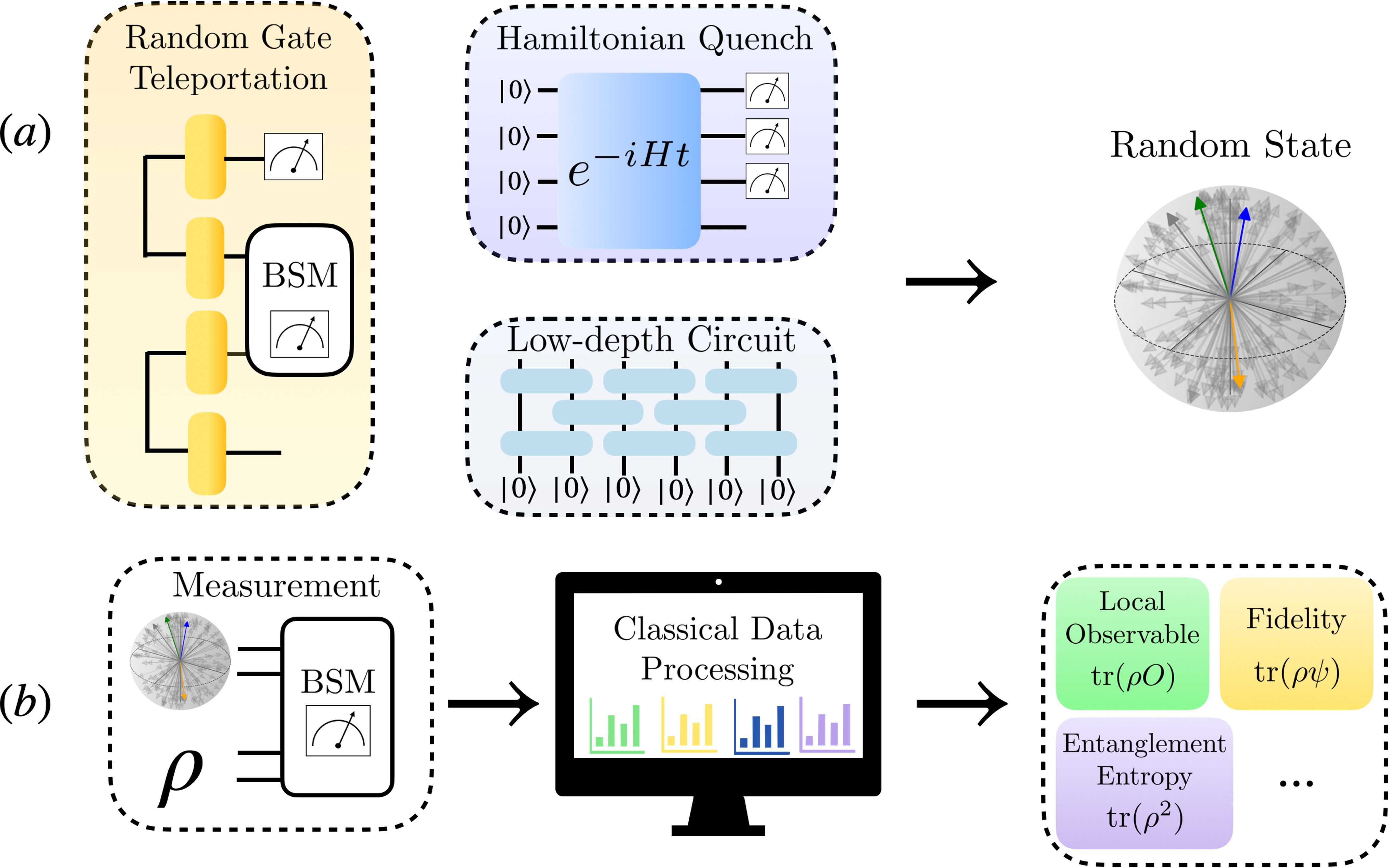}
		\caption{The ancilla-assisted shadow tomography protocol. (a) This protocol requires ancillary random states. The choice of the random state offers a large degree of freedom, providing the protocol with significant generality. For example, the states can be prepared through random gate teleportation proposed in this work, requiring smaller circuit depth, or via the projected ensemble of Hamiltonian evolutions \cite{Cotler_2023}, potentially applicable in analog quantum simulators. The shallow classical shadow \cite{bertoni2023shallow, hu2022classical} can also be adopted into this protocol by preparing the random states via low-depth random circuits. (b) The protocol acquires classical data by performing Bell state measurements (BSM) on the input state $\rho$ and ancillary random states. The measurement data can later be utilized to predict many properties of $\rho$ through classical data processing, such as the expectation values of observables, state fidelity, and nonlinear functions like the second-order R{\'e}nyi entropy.}
		\label{fig:mbcs}
	\end{figure}
	
	Although reconstructing the state requires exponentially many experiments, our primary motivation is to measure properties instead of fully recovering the state $\rho$. In Appendix \ref{app:shadow_analysis}, we provide further analysis of our protocol. We use measurement results to construct unbiased estimators $\hat{\sigma}$ of $\rho$, allowing us to obtain an unbiased estimator $\tr(O \hat{\sigma})$ of the expectation value $\tr(O \rho)$. We show that when the state ensemble is selected as a state $3$-design or as local random states, our protocol achieves performance comparable to that of the classical shadow protocol.
	
	\begin{theorem}\label{thm:shadow_performance}
		To predict the expectation value of $M$ observables $\tr(O_1 \rho), \tr(O_2 \rho), \cdots, \tr(O_M \rho)$ to additive error $\epsilon$, the ancilla-assisted shadow tomography protocol requires $N$ copies of input state $\rho$, where:
		\begin{enumerate}
			\item {$N = \cO(\frac{\log M }{\epsilon^2}\max_i \tr(O_i^2))$ when $\cS$ is chosen as a state $3$-design.}
			\item {$N = \cO(\frac{\log M }{\epsilon^2}\max_i 4^{k_i})$ when $\cS$ is chosen as ensemble of local random states $\ket{\phi} = \ket{\phi_1}\otimes \ket{\phi_2}\otimes \cdots \otimes \ket{\phi_n}$, where each $\ket{\phi_i}$ is uniformly drawn from 
				\begin{equation*}
					\{\ket{0},\ket{+} = \frac{1}{\sqrt{2}}(\ket{0} + \ket{1}),\ket{+i} = \frac{1}{\sqrt{2}}(\ket{0} + i\ket{1})\}.
				\end{equation*} Here, $k_i$ is the locality of observable $O_i$, and $\max_i \norm{O_i}_{\infty} \le 1$.}
		\end{enumerate}
	\end{theorem}
	
	Moreover, we can predict nonlinear functions $f(\rho)$, such as the second-order R{\'e}nyi entropy $\tr(\rho^2)$, by utilizing the U statistics to construct unbiased estimators of $\rho^{\otimes k}$ via $N$ independent unbiased estimator $\{\hat{\sigma}_i\}$ of $\rho$ \cite{Huang_2020, Tran_2023, Hoeffding48Ustatistics}:
	\begin{equation}
		\binom{N}{k}^{-1} \sum_{1 \le i_1,\cdots,i_k \le N} \hat{\sigma}_{i_1} \otimes \hat{\sigma}_{i_2} \otimes \cdots \otimes \hat{\sigma}_{i_k}.
	\end{equation}
	
	A key benefit of our protocol is that it eliminates the need to evolve the input states online, making it particularly suitable for practical scenarios where preparing random states is more accessible than evolving the input state by random unitaries. 
	For example, it has been shown that state designs can be constructed using unitaries that do not form corresponding unitary designs \cite{west2024randomensembles}, suggesting that the circuit complexity of implementing the ancilla-assisted classical shadow may be lower than that of directly evolving the state with random unitaries. Generally speaking, it is practically beneficial to delegate the hardness of implementing dynamics online to offline state preparation, an insight that underlies many important quantum computing schemes including MBQC and magic state distillation \cite{Bravyi_2005}. 
	Moreover, when selecting $\mathcal{S}$ as a global random state ensemble, the state $\ket{\phi} \in \mathcal{S}$ can be prepared using random or Clifford circuits through the random gate teleportation protocol, which reduces the circuit depth required to predict global properties, such as fidelity to target states, to a constant. In contrast, previous work has only achieved logarithmic depth for similar tasks \cite{bertoni2023shallow, hu2022classical, schuster2024randomunitariesextremelylow}.

	Additionally, our protocol exhibits substantial flexibility in selecting ancillary state ensembles, which can be easily adapted to different variations of shadow tomography, such as shallow classical shadows \cite{bertoni2023shallow, hu2022classical}, Hamiltonian-driven classical shadow \cite{Tran_2023, liu2024predicting} and thrifty shadow estimation protocol \cite{Helsen_2023, Zhou_2023}.  Recent studies have demonstrated that random states can be prepared by measuring subsystems of a state evolved under Hamiltonian evolutions \cite{Cotler_2023, Choi_2023}. This state preparation method can be utilized in our protocol to simplify experimental control further and can be implemented in current analog quantum simulators. 

	\section{Addtional analysis for the ancilla-assisted shadow tomography}\label{app:shadow_analysis}
	In this section, we delve deeper into the analysis of the ancilla-assisted shadow tomography. Initially, we examine the POVM operators associated with a chosen state ensemble and analyze their tomographical completeness. Subsequently, we detail the data processing schemes employed in our protocol. We demonstrate that our approach achieves performance comparable to the classical shadow protocol \cite{Huang_2020}, particularly when the state ensemble is selected as state 3-designs or product random states.

	\subsection{POVM of a state ensemble}
	First, we analyze the POVM operators in the ancilla-assisted shadow tomography protocol. For a state ensemble $\cS$ and a given input state $\rho$ in system $B = B_1B_2\cdots B_n$, we select a state $\ket{\phi} \in \cS$ in system $A = A_1A_2\cdots A_n$ and perform a Bell state measurement on each pair of qubits $A_iB_i$.  Let $a_i$ and $b_i$ denote the measurement results on the $i$-th qubit. Denote the unnormalized maximally entangled state as $\ket{\Phi}$, according to Eq.~\eqref{eq:Bell_state} and Eq.~\eqref{eq:post_measurement}, the probability of obtaining bitstrings $\ba = a_1a_2\cdots a_n$ and $\bb = b_1b_2\cdots b_n$ is:
	
	\begin{equation}\label{eq:ab_prob}
		\begin{split}
			p^{\phi}_{\ba\bb} &= \frac{1}{2^n} \tr\{  [X^{\ba}_AZ^{\bb}_A(\ketbra{\Phi})^{\otimes n}Z^{\bb}_A X^{\ba}_A]  \ketbra{\phi} \otimes \rho\}\\
			&= \frac{1}{2^n} \bra{\phi^*} X^{\ba}Z^{\bb}  \rho Z^{\bb} X^{\ba} \ket{\phi^*},
		\end{split}
	\end{equation}
	where $\ket{\phi^*}$ denotes the complex conjugate of $\ket{\phi}$. Multiplying by the probability $p_{\phi}$ of choosen state $\phi$ in $\cS$, this measurement result corresponds to a POVM operator 
	\begin{equation}
		M_{\phi,P} = \frac{p_{\phi}}{2^n} P \ketbra{\phi^*} P^{\dagger},
	\end{equation}
	where $P = X^{\ba}Z^{\bb}, \ba, \bb \in \{0,1\}^n$. Therefore, for a given ensemble $\cS$, the corresponding POVM operators are $\{M_{\phi,P}\}$.
	
	\subsection{Tomographical completeness}
	
	As long as these POVM operators are \emph{tomographically complete}, it is possible to reconstruct the state and thus predict arbitrary properties of the state from the measurement results. Tomographic completeness is guaranteed if and only if for any two arbitrary states $\rho$ and $\sigma$, there exists a Pauli string $P = X^{\mathbf{a}}Z^{\mathbf{b}}$ and a state $\ket{\phi} \in \mathcal{S}$ such that:
	\begin{equation}\label{eq:tomo_complete}
		\begin{split}
			\bra{\phi^*}P^{\dagger} \rho P \ket{\phi^*} \neq \bra{\phi^*}P^{\dagger} \sigma P \ket{\phi^*} \\\Rightarrow  \bra{\phi} P^{\dagger} (\delta \rho)^* P \ket{\phi}  \neq 0,
		\end{split}
	\end{equation}
	where $\delta \rho = \rho - \sigma$. Notice that $\delta \rho$ can be an arbitrary traceless Hermitian matrix in $\mathbb{H}_{2^n}$ up to a multiplicative factor. This characteristic gives the condition for the tomographical completeness of a state ensemble $\cS$:
	\begin{theorem}\label{thm:completeness}
		For a state ensemble $\cS$, the corresponding POVM operators are tomographically complete if and only if the state ensemble $\cS'$ spans the space $\mathbb{H}_{2^n}$ of traceless Hermitian matrices, where
		\begin{equation}
			\cS' = \{\ketbra{\psi} : \ket{\psi} = P \ket{\phi}, \ket{\phi} \in \cS, P = X^{\ba} Z^{\bb}, \ba, \bb \in \{0,1\}^n \}.
		\end{equation}
	\end{theorem} 
	
	After choosing a tomographically complete state ensemble, we can recover the state $\rho$ from the measurement result \cite{Huang_2020, Tran_2023}. Concretely, the POVM maps the state $\rho$ to the distribution of measurement outcomes $P$ via a linear transformation $M$:
	\begin{equation}
		|\rho ) = \begin{pmatrix}
			\rho_{1,1} \\
			\rho_{1,2} \\
			\vdots \\
			\rho_{d,d}
		\end{pmatrix}  \xrightarrow{M} |P) = \begin{pmatrix}
			P_1 \\
			P_2 \\
			\vdots 
		\end{pmatrix}.
	\end{equation}  
	
	Due to tomographic completeness, the linear transformation $M$ has a left inverse $R$. For example, one can choose the Moore--Penrose pseudo-inverse:
	\begin{equation}
		R_{\text{MP}} = (M^{\dagger} M)^{-1} M^{\dagger}. 
	\end{equation}
	To recover all the information in $\rho$, one can perform enough measurements to estimate $|P)$ and apply the recovering map $R$:
	\begin{equation}
		|\rho ) =  R |P).
	\end{equation}  
	
	Although reconstructing the state might be expensive in sample complexity, our primary motivation is to measure properties instead of obtaining all the information about the state. In this case, one can predict some properties of $\rho$ without fully recovering the state. Suppose we already perform $M$ experiment and got measurement result $|P_1), |P_2),\cdots |P_M)$, where $|P_i)$ has a single `1' entry corresponding to the measurement result. The empirical unbiased estimator of $|P)$ is
	\begin{equation}
		|\hat{P}) = \frac{1}{M} \sum_{i=1}^M |P_i).
	\end{equation}
	To predict a given observable $O$, we write it in the vector from $|o)$ and calculate the unbiased estimator of $O$ as
	\begin{equation}
		( o \,| R |\hat{P}).
	\end{equation}
	
	The left inverse $R$ is not unique, and the Moore-Penrose pseudo-inverse might not be optimal in sample complexity \cite{Tran_2023}. Moreover, the computational complexity for calculating the inverse is exponential in qubit numbers. In practice, we may devise clever methods to reduce the sample and computational complexity for specific state ensemble $\cS$, as demonstrated next. 
	
	\subsection{Case study: state 3-design and product random states}
	
	Here, we focus on two state ensembles: global random states and product random states. First, we introduce the classical data processing scheme when $\cS$ is a state 3-design and show that our protocol has equivalent performance to the classical shadow \cite{Huang_2020} using global random unitaries. 
	
	Suppose the Bell state measurement is performed on $\ket{\phi} \otimes \rho$, where $\ket{\phi}$ is drawn from a state 3-design $\cS$, yielding measurement results $a_i, b_i \in \{0,1\}$ for each qubit $i$. Let $\ba = a_1a_2\cdots a_n$ , $\bb = b_1b_2\cdots b_n$ and $\ket{b} = Z^{\bb} X^{\ba} \ket{\phi^*}$. The protocol records the classical data as
	\begin{equation}
		\hat{\sigma} = (2^n+1) \ketbra{b} - I_{2^n}.
	\end{equation} 
	
	This process is repeated $N$ times, resulting in classical data $\{\hat{\sigma}_1, \hat{\sigma}_2, \cdots, \hat{\sigma}_N\}$. As shown in Section \ref{app:shadow_performance}, $\hat{\sigma}_i$s are unbiased estimators of $\rho$. Hence, unbiased estimators of $\mathrm{tr}(O\rho)$ can be obtained. To estimate the expectation value within the desired precision, the median-of-means technique proposed in Ref.~\cite{Huang_2020} is applied. First, new classical data $\hat{\sigma}_{(k)}$ is calculated as follows:
	\begin{equation}
		\hat{\sigma}_{(k)} = \frac{1}{\lfloor N / K \rfloor} \sum_{l=(k-1)\lfloor N / K \rfloor + 1}^{k \lfloor N / K \rfloor} \hat{\sigma}_l.
	\end{equation}
	
	For a given observable $O$, the prediction of $\tr(\rho O)$ is then made by
	\begin{equation}
		\hat{o} = \text{median}\{\mathrm{tr}(O\hat{\sigma}_{(1)}),\mathrm{tr}(O\hat{\sigma}_{(2)}), \cdots, \mathrm{tr}(O\hat{\sigma}_{(K)}) \}.
	\end{equation}
	
	The ancilla-assisted shadow tomography protocol based on state 3-design $\cS$ is summarized in Box~\ref{box:mbcs_protocol_global}.
	
	\begin{mybox}[label={box:mbcs_protocol_global}]{Ancilla-assisted shadow tomography based on state $3$-design}
		\textbf{Input:} 
		\begin{enumerate}
			\item {$N$ copies of an $n$-qubit state $\rho$. }
			\item {Classical description of $M$ observables $O_1, O_2, \cdots, O_M$.}
		\end{enumerate}
		\textbf{Protocol:}
		\begin{enumerate}
			\item For each copy of $\rho$, randomly draw a state $\ket{\phi} \in \cS$, where $\cS$ should be quantum state $3$-design. 
			\item {Perform Bell state measurement on each pair of qubits of $\ket{\phi} \otimes \rho$, yielding measurement results $a_i,b_i \in \{0,1\}$ for $1 \leq i \leq n$. Let $\mathbf{a} = a_1a_2\cdots a_n$, $\mathbf{b} = b_1b_2\cdots b_n$, and $\ket{b} = Z^{\mathbf{b}} X^{\mathbf{a}} \ket{\phi^*}$. Record the classical data 
				\begin{equation}
					\hat{\sigma} = (2^n + 1) \ketbra{b} - I.		
				\end{equation}
			}
			\item  {Obtain $N$ data points $\{\hat{\sigma}_1, \hat{\sigma}_2,\cdots,\hat{\sigma}_N\}$ by repeating Steps 1 to Step 2 on $N$ copies of $\rho$.}
			\item {Split the data into $K$ equally-sized parts, and set
				\begin{equation}
					\hat{\sigma}_{(k)} = \frac{1}{\lfloor N / K \rfloor} \sum_{l=(k-1)\lfloor N / K \rfloor + 1}^{k \lfloor N / K \rfloor} \hat{\sigma}_l.
				\end{equation}
			}
			\item {Output the estimation of $\tr(O_i \rho)$ as:
				\begin{equation}
					\hat{o}_i = \text{median}\{\tr(O_i \hat{\sigma}_{(1)}), \tr(O_i \hat{\sigma}_{(2)}), \cdots, \tr(O_i \hat{\sigma}_{(n)})\}.
			\end{equation}}
		\end{enumerate}
	\end{mybox}

	As proved in Section \ref{app:shadow_performance}, this scheme exhibits equivalent performance to the original random Clifford measurement.
	
	\begin{proposition}\label{prop:shadow_global}
		Ancilla-assisted shadow tomography protocol based on state 3-design depicted in Box~\ref{box:mbcs_protocol_global} can predict the expectation value of $M$ observables $\tr(O_1 \rho), \tr(O_2 \rho), \cdots, \tr(O_M \rho)$ to additive error $\epsilon$, provided that $N \ge \cO(\frac{\log M }{\epsilon^2}\max_i \tr(O_i^2))$. 
	\end{proposition}
	
	We can also choose $\cS$ as the tensor product of local random states. This state ensemble is well-suited for predicting local observables. 
	\begin{proposition}\label{prop:shadow_local}
		Ancilla-assisted shadow tomography using random states $\ket{\phi} = \ket{\phi_1}\otimes \ket{\phi_2}\otimes \cdots \otimes \ket{\phi_n}$, where each $\ket{\phi_i}$ is uniformly drawn from $\{\ket{0},\ket{+} = \frac{1}{\sqrt{2}}(\ket{0} + \ket{1}),\ket{+i} = \frac{1}{\sqrt{2}}(\ket{0} + i\ket{1}\}$, can predict the expectation value of $M$ arbitrary $k$-local observables $\tr(O_1 \rho), \tr(O_2 \rho), \cdots, \tr(O_M \rho)$ that satisfies $\max_i \norm{O_i}_{\infty} \le 1$ to additive error $\epsilon$, provided that $N \ge \cO(\frac{\log M }{\epsilon^2}4^k)$.
	\end{proposition}
	To prove this proposition, note that the protocol is equivalent to measuring each qubit of the state $\rho$ with six states:
	\begin{equation}
		\text{stab}_1 = \{\ketbra{0}, \ketbra{1}, \ketbra{\pm}, \ketbra{\pm i}\}.
	\end{equation}
	This is exactly the classical shadow protocol based on local random unitaries so this result can be derived by Proposition S3 in Ref.~\cite{Huang_2020} and follows the same data postprocessing scheme. Suppose the measurement result on the $i$-th qubit is $\ket{\psi_i} \in \text{stab}_1$, then, the classical data is recorded as:
	\begin{equation}
		\hat{\sigma} = \bigotimes_{i=1}^n\hat{\psi}_i,\quad \hat{\psi}_i = 3\ketbra{\psi_i} - I.
	\end{equation}
	The rest of classical postprocessing is the same as in Box \ref{box:mbcs_protocol_global}. This method is computationally efficient when the observables are local. Theorem \ref{thm:shadow_performance} can be proved by combining Proposition \ref{prop:shadow_global} and Proposition \ref{prop:shadow_local}. 
	
	Moreover, the method in Ref.~\cite{Huang_2020} can be adopted to predict nonlinear functions $f(\rho)$. Given independent and unbiased estimators $\{\hat{\sigma}_1, \hat{\sigma}_2,\cdots,\hat{\sigma}_N\}$, an unbiased estimators $\rho \otimes \rho$ of can be constructed as follows:
	\begin{equation}
		\hat{\mu}_2 = \frac{1}{N(N-1)} \sum_{i \neq j} \hat{\sigma}_i \otimes \hat{\sigma}_j.
	\end{equation}
	A nonlinear function $\tr(O \rho \otimes \rho)$ can be predicted by calculating $\tr(O \hat{\mu_2})$. This process can be repeated multiple times, and taking the median of these repetitions can reduce the prediction error. This allows for estimating nonlinear properties like the second R{\'e}nyi entropy. Although the sample complexity remains exponential, there is a considerable reduction compared to brute-force methods such as full-state tomography. Additionally, this process can be generalized to higher moments of $\rho$ by constructing unbiased estimators of $\rho^{\otimes k}$ via the U statistics \cite{Huang_2020, Tran_2023, Hoeffding48Ustatistics}:
	\begin{equation}
		\hat{\mu}_k = \binom{N}{k}^{-1} \sum_{1 \le i_1,\cdots,i_k \le N} \hat{\sigma}_{i_1} \otimes \hat{\sigma}_{i_2} \otimes \cdots \otimes \hat{\sigma}_{i_k}.
	\end{equation}

	\subsection{Performance analysis of state 3-design}\label{app:shadow_performance} 
	Here, we analyze the ancilla-assisted shadow tomography protocol based on state $3$-design. Our analysis mainly follows the approach outlined in Ref.~\cite{Huang_2020}. Suppose we perform Bell state measurement and get measurement result $\ba, \bb$. Let $\ket{\phi,\ba\bb} =Z^{\bb} X^{\ba} \ket{\phi^*}$. According to Eq.~\eqref{eq:ab_prob}, the expectation value of $\ketbra{\phi,\ba\bb}$ is given by 
	\begin{equation}\label{eq:data_mean}
		\begin{split}
			\mathbb{E}_{\phi,\ba\bb} \ketbra{\phi,\ba\bb} &= \mathbb{E}_{\phi} \sum_{\ba\bb}p^{\phi}_{\ba\bb} \ketbra{\phi,\ba\bb}\\
			&= \frac{1}{2^n}\mathbb{E}_{\phi} \sum_{\ba\bb}  \bra{\phi^*} X^{\ba}Z^{\bb}  \rho Z^{\bb} X^{\ba} \ket{\phi^*} Z^{\bb}X^{\ba}\ketbra{\phi^*} X^{\ba} Z^{\bb}  \\
			&= \frac{1}{2^n} \sum_{\ba\bb} \mathbb{E}_{\phi} \bra{\phi^*} X^{\ba}Z^{\bb}  \rho Z^{\bb} X^{\ba} \ket{\phi^*} Z^{\bb}X^{\ba}\ketbra{\phi^*} X^{\ba} Z^{\bb}  \\
			&= \frac{1}{2^n} \sum_{\ba\bb} \frac{1}{2^n} \mathcal{D}_{1/(2^n + 1)}(\rho) \\
			&= \mathcal{D}_{1/(2^n + 1)}(\rho).
		\end{split}
	\end{equation}
	where $\mathcal{D}_p(\rho) = p \rho + (1-p) \frac{I}{2^n}$. In the second equation, we use the equality 
	\begin{equation}
		p^{\phi}_{\ba\bb} = \tr(\ketbra{X^{\ba}Z^{\bb},I} [\ketbra{\phi} \otimes \rho]) = \frac{1}{2^n}\bra{\phi^*} X^{\ba}Z^{\bb}  \rho Z^{\bb} X^{\ba} \ket{\phi^*}.
	\end{equation}
	In the third equation, we swap the order of summation. The fourth equation leverages the $3$-design property of $\cS$. For the Hermitian matrix in $\mathbb{H}_{2^n}$ and a Pauli string $P =  Z^{\bb} X^{\ba}$, we have:
	\begin{align}
		\mathbb{E}_{\phi}
		P\ketbra{\phi^*}P^{\dagger} \bra{\phi^*} P^{\dagger} A P\ket{\phi^*}
		=& \frac{A+\mathrm{tr}(A)I}{(2^n+1)2^n} = \frac{1}{2^n}\mathcal{D}_{1/(2^n+1)}(A)
		\; \; \textrm{for $A \in \mathbb{H}_{2^n}$,}   \\
		\mathbb{E}_{\phi}
		P\ketbra{\phi^*}P^{\dagger}
		\bra{\phi^*} P^{\dagger} B_0 P \ket{\phi^*} \bra{\phi^*} P^{\dagger} C_0 P\ket{\phi^*}  =&
		\frac{\mathrm{tr}(B_0C_0)I +B_0C_0 +C_0B_0}{(2^n+2)(2^n+1)2^n}
		\;\;\text{for $B_0,C_0\in \mathbb{H}_{2^n}$ traceless.} 
	\end{align}
	
	Now, we analyze the estimator $\hat{o}=\tr(O \hat{\sigma})$. The expectation of the data $\hat{\sigma}$ satisfies 
	\begin{equation}
		\begin{split}
			\mathbb{E} \hat{\sigma} &= \mathbb{E}_{\phi,\ba\bb}[ (2^n + 1)\ketbra{\phi,\ba\bb} - I ]\\
			&=  (2^n + 1)\mathbb{E}_{\phi,\ba\bb}[ \ketbra{\phi,\ba\bb} ] - I \\
			&= (2^n + 1) \mathcal{D}_{1/(2^n + 1)}(\rho) - I \\
			&= \rho.
		\end{split}
	\end{equation}
	Hence, 
	\begin{equation}
		\mathbb{E} \hat{o} = \mathbb{E} \tr(O \hat{\sigma}) = \tr(O \mathbb{E} \hat{\sigma}) = \tr(O \rho).
	\end{equation}
	That is, $\hat{o}$ is an unbiased estimator of $\tr(O\rho)$. Next, we analyze the variance of $\hat{o}$. Define the linear map $\mathcal{M} = \mathcal{D}_{1 / (2^n + 1)}$, according to Lemma S1 in Ref.~\cite{Huang_2020}, we have
	\begin{equation}\label{eq:var_o}
		\text{Var}[\hat{o}] = \mathbb{E}[(\hat{o} - \mathbb{E}[\hat{o}] )^2] \le ||O_0||_{shadow},
	\end{equation}
	where $O_0 = O - \frac{\tr(O)}{2^n} I$, and the shadow norm is defined as
	\begin{equation}\label{eq:def_shadown_norm}
		||O||_{shadow} = \max_{\sigma: \text{ state}} (\mathbb{E}_{\phi} \sum_{\ba\bb} p_{\ba\bb}^{\phi} \bra{\phi,\ba\bb} \mathcal{M}^{-1}(O)\ket{\phi,\ba\bb}^2)^{1/2}.
	\end{equation}
	
	Here, we establish the bound of shadow norm for the ancilla-assisted shadow tomography.
	\begin{proposition}[Shadow norm for ancilla-assisted shadow tomography]
		For any observable $O$, its traceless part $O_0 = O - \frac{\tr(O)}{2^n} I$ satisfies
		\begin{equation}\label{eq:bound_shadow_norm}
			||O_0||^2_{shadow} \le 3 \tr(O^2).
		\end{equation}
	\end{proposition}
	\begin{proof}
		Following Eq.~(S42) in Ref.~\cite{Huang_2020}, we have
		\begin{equation}
			\mathcal{M}^{-1}(O_0) = (2^n + 1) O_0
		\end{equation}
		for any traceless $O_0 \in \mathbb{H}_{2^n}$. Then, from Eq.~\eqref{eq:data_mean} and Eq.~\eqref{eq:def_shadown_norm}, we have
		
		\begin{align}
			\| O_0 \|_{\mathrm{shadow}}^2 =& \max_{\sigma: \text{ state}}  \big(\frac{1}{2^n}\mathbb{E}_{\phi}\sum_{\ba \bb}  \bra{\phi^*} X^{\ba}Z^{\bb}  \sigma Z^{\bb} X^{\ba} \ket{\phi^*}  [\bra{\phi^*} X^{\ba}Z^{\bb} (2^n+1)O_0 Z^{\bb} X^{\ba} \ket{\phi^*} ]^2 \big) \nonumber \\
			=& \max_{\sigma: \text{ state}} \mathrm{tr} \big( \sigma\; \frac{1}{2^n}\sum_{\ba\bb} \mathbb{E}_{\phi} Z^{\bb} X^{\ba} \ketbra{\phi^*}X^{\ba}Z^{\bb} [\bra{\phi^*}X^{\ba}Z^{\bb} (2^n+1)O_0 Z^{\bb}X^{\ba}\ket{\phi^*}] ^2 \big) \nonumber \\
			=& \max_{\sigma: \text{ state}} \mathrm{tr} \left( \sigma \;2^n \frac{(2^n+1)^2 \left(\mathrm{tr}(O_0^2)I + 2 O_0^2\right)}{(2^n+2)(2^n+1)2^n}\right)
			= \frac{2^n+1}{2^n+2} \max_{\sigma \text{ state}} \left( \mathrm{tr}(\sigma) \mathrm{tr}(O_0^2) + 2 \mathrm{tr} \left( \sigma O_0^2 \right) \right).
		\end{align}
		Note that $\tr(\sigma O_0^2) \le ||O_0^2||_{\infty} \le \tr(O_0^2)$, $\tr(\sigma) = 1$ and $\tr(O_0^2) = \tr(O^2) - \frac{\tr(O^2)}{2^n} \le \tr(O^2)$. Hence, we obtain Eq.~\eqref{eq:bound_shadow_norm}.
	\end{proof}
	
	After obtaining the shadow norm of the operator, we bound the variance of $\hat{o}$ according to Eq.~\eqref{eq:var_o}:
	\begin{equation}
		\mathrm{Var}[\hat{o}] \le 3 \mathrm{tr}(O^2).
	\end{equation}
	
	Directly applying concentration inequalities to $\hat{o}$ is not feasible because $\hat{o}$ and its higher moments are not yet bounded. To address this, we employ the median-of-means method, following Ref.~\cite{Huang_2020}. Firstly, we average $B$ estimators to obtain $\hat{o}_{(l)} = \hat{o}_{(l-1)B+1},\hat{o}_{(l-1)B + 2},\cdots,\hat{o}_{lB}$, where $\hat{o}_i$ are independent and identically distributed estimators. The estimator $\hat{o}_{(l)}$ remains unbiased, and its variance is suppressed by $B$ through standard calculations, yielding:
	\begin{equation}
		\mathrm{Var}[\hat{o}_{(l)}] \le \frac{3 \mathrm{tr}(O^2)}{B}.
	\end{equation}
	Setting $B = \frac{30 \mathrm{tr}(O^2)}{\epsilon^2}$, we ensure $\mathrm{Var}[\hat{o}_{(l)}] \le \frac{\epsilon^2}{10}$. By Markov's inequality, we have:
	\begin{equation}
		\mathrm{Pr}[|\hat{o}_{(l)} - \mathrm{tr}(O\rho)| > \epsilon] \le \frac{\mathrm{Var} [\hat{o}_{(l)}]}{\epsilon^2} \le \frac{1}{10}.
	\end{equation}
	
	Now, we apply Hoeffding's inequality to the indicator function $\mathbbm{1}_{|\hat{o}_{(i)} - \tr(O \rho)| > \epsilon}$. After calculating the median $m$ of ${\hat{o}_{(1)}, \hat{o}_{(2)}, \ldots, \hat{o}_{(K)}}$, the probability that $|m-\tr(O\rho)| > \epsilon$ is equal to the probability that 
	\begin{equation}
		\frac{1}{K}\sum_{i=1}^K \mathbbm{1}_{|\hat{o}_{(i)} - \tr(O \rho)| > \epsilon} \ge \frac{1}{2}.
	\end{equation}
	According to Hoeffding's inequality, this probability will be $\exp(-\cO(K))$.
	
	Given $M$ observables, by setting $B = \frac{30 \max_i \mathrm{tr}(O_i^2)}{\epsilon^2}$ and $K = \cO(\log \frac{M}{\delta})$, the probability that there exists an estimation $m_i$ of $\tr(O_i \rho)$ such that $|m_i - \mathrm{tr}(O_i \rho)| > \epsilon$ is at most
	\begin{equation}
		M \exp(-\cO(K)) \le \delta
	\end{equation}
	by a union bound. Hence, to estimate any $\mathrm{tr}(O_i \rho)$ up to an additive error of $\epsilon$ with probability $\delta$, choosing $N = BK = O\left(\frac{\log \frac{M}{\delta} }{\epsilon^2}\max_i \mathrm{tr}(O_i^2)\right)$ suffices.

\end{document}